\documentclass[11pt,letterpaper]{article}
\usepackage[T1]{fontenc}
\usepackage[utf8]{inputenc}
\usepackage{mathpazo}
\usepackage[scaled=.94]{helvet}
\usepackage{amsmath,amssymb,amsthm,mathtools}
\usepackage[margin=0.88in,top=0.82in,bottom=0.85in,headheight=23pt,headsep=18pt]{geometry}
\usepackage{microtype}
\usepackage{pgfplots}
\pgfplotsset{compat=1.18}
\usepackage{xcolor,booktabs,tabularx,array,enumitem}
\usepackage{fancyhdr,titlesec,needspace,float}
\usepackage[colorlinks=true,linkcolor=navy,citecolor=gold,urlcolor=gold]{hyperref}
\usepackage[nameinlink,noabbrev,capitalise]{cleveref}

\definecolor{navy}{HTML}{1C2843}
\definecolor{gold}{HTML}{AE7100}
\definecolor{muted}{HTML}{657083}
\definecolor{rulegray}{HTML}{DCE1E8}
\hypersetup{pdftitle={Cities of Signals: Compression, Separation, and the Geometry of Novelty},pdfauthor={Marc Nunes},pdfsubject={Research Note, Version 17, October 2026},pdfkeywords={signal compression, directional statistics, novelty, spherical geometry, mean resultant length}}
\renewcommand{\headrulewidth}{0.4pt}
\renewcommand{\headrule}{\hbox to\headwidth{\color{rulegray}\leaders\hrule height\headrulewidth\hfill}}
\titleformat{\section}{\large\bfseries\color{navy}}{\thesection.}{.6em}{}
\titleformat{\subsection}{\normalsize\bfseries\color{navy}}{\thesubsection.}{.6em}{}
\titlespacing*{\section}{0pt}{2.3ex plus .7ex minus .3ex}{1ex}
\titlespacing*{\subsection}{0pt}{1.8ex plus .5ex minus .2ex}{.7ex}
\setlist{itemsep=2pt,topsep=4pt,leftmargin=1.7em}

\allowdisplaybreaks[1]
\newtheoremstyle{paperplain}{8pt}{7pt}{\itshape}{}{\bfseries}{.}{.5em}{}
\theoremstyle{paperplain}
\newtheorem{theorem}{Theorem}[section]
\newtheorem{proposition}[theorem]{Proposition}
\newtheorem{corollary}[theorem]{Corollary}

\theoremstyle{definition}
\newtheorem{definition}[theorem]{Definition}
\newtheorem{example}[theorem]{Example}
\theoremstyle{remark}
\newtheorem{remark}[theorem]{Remark}

\newcommand{\R}{\mathbb R}
\newcommand{\E}{\mathbb E}
\newcommand{\Prob}{\mathbb P}
\newcommand{\sphere}{\mathbb S}
\newcommand{\ip}[2]{\langle #1,#2\rangle}
\newcommand{\norm}[1]{\lVert #1\rVert}
\newcommand{\IC}{\operatorname{IC}}
\newcommand{\Var}{\operatorname{Var}}
\newcommand{\Cov}{\operatorname{Cov}}
\newcommand{\tr}{\operatorname{tr}}
\newcommand{\id}{\mathrm I}
\newcommand{\dto}{\mathrel{\Rightarrow}}
\newcommand{\doi}[1]{\href{https://doi.org/#1}{doi:#1}}

\begin{document}
\hypersetup{pageanchor=false}
\begin{titlepage}
\thispagestyle{empty}
\vspace*{0.80in}
\noindent{\sffamily\bfseries\color{navy}\fontsize{27}{32}\selectfont Cities of Signals\par}
\vspace{0.16in}
\noindent{\sffamily\color{navy}\fontsize{14}{19}\selectfont Compression, Separation, and the Geometry of Novelty\par}
\vspace{0.28in}
\noindent{\color{gold}\rule{1.48in}{1.4pt}}
\vspace{0.27in}

\noindent Marc Nunes\par
\smallskip
\noindent{\small AlphaNova\quad$\cdot$\quad\href{mailto:marc.nunes@alphanova.tech}{marc.nunes@alphanova.tech}}\par
\vspace{0.18in}
\noindent{\small\color{muted}Research Note\quad$\cdot$\quad Version 17\quad$\cdot$\quad Revised 2 October 2026}

\vspace{0.53in}
\noindent\begin{minipage}{\textwidth}
\small
\textbf{Abstract.}
A cross-sectional signal is a forecast vector over $d$ assets at each date; demeaned and normalized, it is a point on a sphere. Its \emph{city} is the direction of its time-averaged vector in a common target-aligned frame: a compressed summary that can be published in place of the full history. We ask when the angle between two cities is a conservative estimate of the angle between the two histories. For independent, uniformly distributed histories, the probability that the history angle is at least the city angle tends to one when the city angle is at most $90^\circ$, with an $O(T^{-1/2})$ error bound that is uniform in dimension, where $T$ is the history length. Selecting histories for a positive average information coefficient (IC) changes the unconditional limit to a constant $\gamma_q>1/2$, approximately $56\%$ for 20 assets. The ordering bound extends to independent, rotationally symmetric histories with temporal dependence, and the conditional limit survives when the products of the two histories' lag correlations are summable. A separate stationary ergodic theorem allows the two signals to share movement and identifies the population margin that determines the limiting ordering. Empirically, among 266,815 positive-IC pairs in AlphaNova's Competition May 2026, the history angle is at least the city angle for $58.9\%$ of pairs overall, for $74.1\%$ at city angles up to $90^\circ$, and for $1.5\%$ beyond $90^\circ$. These frequencies describe the observed library; they do not validate the uniform null. Residual agreement is predominantly positive, and an empirical second-moment diagnostic cannot be matched within that null by adjusting temporal dependence alone. As a screen on the full 776-signal sample, a city angle of at least $60^\circ$ selects pairs with a $94.3\%$ success rate at the same history-angle threshold, against a base rate of $85.0\%$, while retaining $72.5\%$ of pairs. On the deterministic side, explicit constructions show what city directions alone cannot certify, the retained lengths of the means give sharp bounds, and means over shared time blocks recover the missing information, in the competition library only at fine temporal resolution. Persistent positive IC confines a city and bounds its motion; convergence additionally requires a stable mean direction.
\end{minipage}

\vfill
\noindent{\footnotesize\color{muted}\textbf{Keywords:} cross-sectional signals; compression; novelty; directional statistics; mean resultant length; temporal alignment; residual frame freedom; spherical geometry.}
\end{titlepage}
\hypersetup{pageanchor=true}
\setcounter{page}{1}

\section{The question: what can a city certify?}

A signal library contains histories of cross-sectional forecasts. At each date a signal is a vector over the available assets; after demeaning and normalization, that vector lies on a sphere. Comparing two complete histories requires comparing their forecasts at matched dates. A platform that cannot expose every stored history may instead publish a compressed summary of each signal. The summary studied here is the signal's time-averaged forecast vector, computed in a common target-aligned coordinate frame and rescaled to unit length. We call that direction the signal's \emph{city}.

The practical question is this: if a candidate's city is far from the city of an existing signal, is the candidate's history also far from that signal's history? We study this question one pair at a time; a guarantee that holds simultaneously for every signal in a library is a stronger requirement. A useful answer must keep two distinctions in view. The first is between a guarantee that holds for arbitrary histories and a statement that holds with high probability under a specified temporal model. The second is between signed and absolute correlation. A nearly opposite signal counts as separated under a one-sided correlation cap, but not under a rule that excludes strong correlation of either sign.

We ask two empirical questions about each pair of signals. \emph{Angle preservation} asks whether the angle between the two histories is at least the angle between their cities. \emph{Threshold screening} asks, among pairs whose city angle clears a chosen threshold, how often the history angle clears a specified threshold. When the two thresholds match, preservation implies screening success for an accepted pair. The converse fails: a pair can fall short of preservation and still clear the required history threshold. Both questions are answered on the signal library of AlphaNova's Competition May 2026 (\cref{sec:empirical}). Among pairs in which both signals have a positive average information coefficient (IC, the cross-sectional correlation between signal and realized target), $58.9\%$ preserve separation. The rate is $74.1\%$ when the city angle is at most $90^\circ$ and $1.5\%$ beyond $90^\circ$. For the screening question, on the full sample, a $60^\circ$ city cutoff gives $94.3\%$ success at the same history threshold, against a base rate of $85.0\%$. On the theory side, a uniform-history model, with an extension to serially dependent histories, serves as a reference, and a separate population-margin theorem allows the two signals to share movement. Observing aggregate preservation rates similar to a model's limits does not establish either model's assumptions.

This note follows three related studies. \emph{Signal Correlation, IC, and PnL Dependence} separates target alignment from transverse signal geometry and identifies the residual frame freedom~\cite{nunesIC}. \emph{Large Signal Libraries} studies limits of averages over signal designs~\cite{nunesLarge}. \emph{Separated Signal Libraries} treats correlation constraints as separation conditions on whole histories~\cite{nunesSeparated}. Here we average over time within a single signal, and we compare compressed images of histories rather than the histories themselves. The construction and its motivating questions were introduced in the AlphaNova posts~\cite{blogCities,blogCompression,blogWandering}.

The main observation is that a city discards two kinds of information. Temporal averaging removes fluctuations around the mean; normalization then removes the mean's length. The first loss leaves an unobserved residual cross-moment. The second makes an arbitrarily small mean look like a full-length direction. Together they explain why separated cities can coexist with almost identical histories.

Keeping the length of the mean, called the \emph{mean resultant length}, makes the lost information explicit and yields our first bound. The paper then adds four things: unit-vector examples showing that, under stated dimension conditions, this bound cannot be improved; the probability that a pair preserves separation under a uniform null model; nested intervals obtained from progressively finer time blocks; and deterministic bounds on how the city of a signal with a persistent positive information coefficient can move. \Cref{sec:impossible,sec:sharp} develop the constructions and certificates. \Cref{sec:prob} studies independent and temporally dependent rotationally symmetric histories, random temporal matching, independent-date concentration, and stationary ergodic population margins. \Cref{sec:blocks,sec:stability} develop refinement and dynamics. \Cref{sec:empirical} presents numerical checks and an empirical study of Competition May 2026 signals, and \Cref{sec:conclusion} concludes.

\subsection{Relation to established theory and scope}

The normalized Euclidean mean is a standard object in directional statistics~\cite{mardiajupp}. The same normalized centroid of unit vectors is the concept vector of spherical $k$-means clustering~\cite{dhillon}. Its asymptotic behavior on spheres has been studied without rotational-symmetry assumptions by Hendriks, Landsman, and Ruymgaart~\cite{hendriks}. Bounds on an average inner product minus the inner product of averages belong to the broader family of covariance and Gr\"uss-type inequalities; see, for example, Dragomir~\cite{dragomir1999,dragomir}. The random-matching statistic is a permutation sum of the kind studied by Hoeffding~\cite{hoeffding}. Our independent-date concentration argument is an application of McDiarmid's bounded-differences inequality~\cite{mcdiarmid}.

These ingredients are classical; in particular, the basic residual bound is Cauchy--Schwarz. The contribution is their assembly, together with four additions: explicit unit-vector constructions that attain the bounds; ordering bounds that are uniform in dimension, both for independent dates and for temporally dependent, rotationally symmetric histories (\cref{thm:uniformorderingrate,thm:dependentordering}); a hierarchy of nested intervals; and deterministic city dynamics that separate positive target alignment, confinement, and convergence. The population-margin extension uses the stationary ergodic law of large numbers. The empirical study measures preservation among positive-IC signals and compares threshold screening with its base rate and its retention cost. We do not claim a new general concentration inequality or a new theory of directional means. Nor is the city a similarity-preserving sketch in the sense of random projections~\cite{jl} or rounding-based hashes~\cite{charikar}: those are randomized maps designed to preserve inner products, whereas the city is a deterministic temporal average that discards them unless further information is retained.

Four limits on scope apply throughout. First, all results concern signal correlation in a specified common frame. Aligning realized targets is an ex post diagnostic operation, so the geometric constructions are not a deployable forecasting rule. Second, the certificates are sufficient conditions: when a certificate is inconclusive, the pair may still be distinct. Third, each probabilistic claim holds under its stated sampling law. None of them, taken alone, establishes diversification of profit and loss (PnL), economic value, or stability across future market regimes. Fourth, the empirical study uses historical signal evaluations from AlphaNova's Competition May 2026 and describes the exported signal geometry under its stored frame convention. It establishes neither prospective trading performance nor the sampling assumptions of the uniform null.

\section{Histories, cities, and the common frame}\label{sec:setup}

Let there be $d$ assets and define the neutral space
\begin{equation}
 H=\{s\in\R^d:\mathbf 1^\top s=0\},\qquad q=\dim H=d-1\ge2.
\end{equation}
A demeaned, normalized signal at date $t$ is $S_{i,t}\in\sphere(H)$, where $\sphere(H)$ denotes the unit sphere of $H$. Thus $\sphere(H)\cong\sphere^{q-1}$. As an illustration, with $d=20$ assets, $q=19$ and the sphere is $\sphere^{18}$. Throughout, $d$ remains arbitrary subject to the stated dimension assumptions.

Let $Q_t\in\sphere(H)$ be the normalized realized target direction. Fix a unit vector $e\in H$ as the North Pole and choose an orthogonal map $G_t:H\to H$ with $G_tQ_t=e$. The \emph{same} map is applied to every signal at date $t$. Write
\begin{equation}\label{eq:frame}
 X_{i,t}=G_tS_{i,t},\qquad
 p_{i,t}=\ip{X_{i,t}}e=\ip{S_{i,t}}{Q_t}=\IC_{i,t}.
\end{equation}
Here IC denotes the cross-sectional information coefficient~\cite{grinold,qianhua}: the correlation between the demeaned, normalized signal and target. At each date, we thus express every signal in a coordinate system where the realized target points along the fixed North Pole axis. The signal's coordinate along this axis is its realized IC; its perpendicular component describes its geometry orthogonal to the target. This alignment is constructed after the target is observed. Because the same orthogonal map is applied to every signal, their pairwise angles at that date are preserved.

All subsequent temporal averages use this specified frame convention. Dates are equally weighted unless block weights are explicitly introduced.

\begin{definition}[City and mean resultant length]\label{def:city}
For $T$ dates, set
\begin{equation}\label{eq:city}
 m_i=\frac1T\sum_{t=1}^T X_{i,t},\qquad
 r_i=\norm{m_i},\qquad c_i=\frac{m_i}{r_i}\quad(r_i>0).
\end{equation}
The direction $c_i$ is the city. The scalar $r_i$ is its mean resultant length~\cite{mardiajupp}, also called the retained length below. A city is undefined when $m_i=0$; the mean vector itself remains well-defined.
\end{definition}

The city minimizes average squared chordal distance to the observed vectors; it is the extrinsic sample mean of the history on the sphere~\cite{bhattacharya}. Indeed, for $c\in\sphere(H)$,
\begin{equation}
 \frac1T\sum_t\norm{X_{i,t}-c}^2=2-2\ip{m_i}c.
\end{equation}
The minimizer is uniquely $c_i$ if $r_i>0$, whereas every direction is a minimizer if $r_i=0$.

For a pair of signals define
\begin{equation}\label{eq:correlations}
 D_{ij}=\frac1T\sum_t\ip{X_{i,t}}{X_{j,t}},\qquad
 C_{ij}=\ip{c_i}{c_j}.
\end{equation}
Because $G_t$ is common and orthogonal, $D_{ij}$ equals the original time-averaged cross-sectional signal correlation. Equivalently it is the inner product of the unit histories
\begin{equation}
 x_i=T^{-1/2}(X_{i,1},\ldots,X_{i,T})\in\R^{qT}.
\end{equation}
The history and city angles are
\begin{equation}\label{eq:angles}
 \theta_{ij}=\arccos D_{ij},\qquad \phi_{ij}=\arccos C_{ij}.
\end{equation}
The first is the angle between stacked histories, not the average of the datewise angles. \Cref{tab:notation} collects the recurring quantities.

\begin{table}[H]
\centering\small
\caption{Recurring quantities and the information they retain.}
\label{tab:notation}
\begin{tabularx}{\textwidth}{@{}lX@{}}
\toprule
Quantity & Meaning\\
\midrule
$X_{i,t}$ & Unit signal at one date in the common target-aligned frame\\
$m_i=r_ic_i$ & Temporal mean, retaining both direction and length\\
$r_i$ & Mean resultant length; $r_i^2$ is the energy of the constant component\\
$D_{ij}$ & Matched-date inner product averaged over the history\\
$C_{ij}$ & Inner product of the two city directions\\
$K_{ij}$ & Inner product of temporally centered histories\\
$\kappa$ & Population residual cross-moment under a specified joint signal law\\
\bottomrule
\end{tabularx}
\end{table}

\subsection{The remaining frame choice matters}\label{sec:gauge}

Fixing the target at $e$ leaves a residual orthogonal action on $e^\perp$. This freedom preserves every matched-date comparison but need not preserve temporal averages.

\begin{proposition}[City geometry is not invariant to arbitrary residual frame changes]\label{prop:gauge}
Suppose $R_t$ is orthogonal and $R_te=e$. Replacing $G_t$ by $R_tG_t$ preserves all $D_{ij}$ and all $p_{i,t}$. In general it changes the cities, their retained lengths, and their mutual angles. If $R_t=R$ is constant over time, city angles and retained lengths are preserved.
\end{proposition}
\begin{proof}
The first claims follow from orthogonality and $R_te=e$. For the failure of invariance, take two constant histories
\[
 X_t=ae+bu,\qquad Y_t=ae+bv,
\]
where $a,b>0$, $a^2+b^2=1$, and $u,v\in e^\perp$ are distinct unit vectors. On two dates use the identity and then reflection of the entire transverse space. The transformed means both equal $ae$, so the cities coincide. The unchanged matched-date correlation is $a^2+b^2\ip uv$, which need not be one. If the residual map is constant, it factors out of every temporal average and acts on all cities by the same isometry.
\end{proof}

Theorems about cities therefore require a complete frame convention shared by all signals. One possible choice is a measurable Householder rule that depends on the target; another is a transported frame. The results below hold for any specified convention that satisfies their assumptions. Changing the convention can change the numerical cities, and it can change whether a probabilistic assumption holds. A frame that is estimated or transported using several dates can also introduce dependence between transformed observations. Whenever we assume independent dates below, the assumption is about the transformed observations.

\subsection{Average IC and retained length}

The longitudinal coordinate of the mean gives the exact identity
\begin{equation}\label{eq:ic}
 \overline{\IC}_i=\ip{m_i}e=r_i\ip{c_i}e,
 \qquad |\overline{\IC}_i|\le r_i\le1.
\end{equation}
Thus the city and its retained length together recover average IC. A strictly positive average IC guarantees that the empirical mean is nonzero. It does not keep the retained length bounded away from zero along a sequence of problems, and it does not force the transverse component of the mean to vanish. \Cref{sec:stability} derives the quantitative consequences of a uniform positive IC margin and separates confinement from convergence.

\section{Why city direction alone is insufficient}\label{sec:impossible}

City separation cannot imply history separation for all histories unless more than the city direction is known. This remains true even when every observed IC is positive.

To see what is at stake, suppose a platform publishes only city directions and a candidate's city lies far from every published city. The next theorem constructs pairs for which that appearance is wrong: the cities are any prescribed angle apart, yet the full histories are nearly identical.

\begin{theorem}[Separated cities with almost identical histories]\label{thm:impossible}
Assume $q\ge3$. For every fixed angle $0<\phi<\pi$ there are pairs of unit-vector histories whose cities are exactly $\phi$ apart, whose ICs are strictly positive at every date, and whose history correlations approach one.
\end{theorem}
\begin{proof}
Choose orthonormal vectors $e,u,w\in H$ and set
\[
 c_+=\cos(\phi/2)e+\sin(\phi/2)u,\qquad
 c_-=\cos(\phi/2)e-\sin(\phi/2)u.
\]
For $0<r<1$ and two dates indexed by $\sigma\in\{-1,1\}$, define
\begin{equation}\label{eq:counterexample}
 X_\sigma=rc_++\sigma\sqrt{1-r^2}\,w,\qquad
 Y_\sigma=rc_-+\sigma\sqrt{1-r^2}\,w.
\end{equation}
All observations have unit norm. The means are $rc_+$ and $rc_-$, so the city angle is $\phi$. Every IC equals $r\cos(\phi/2)>0$. Nevertheless,
\begin{equation}\label{eq:counterexampleD}
 D=1-r^2(1-\cos\phi)\longrightarrow1\qquad(r\downarrow0).
\end{equation}
Repeating the two dates gives the same construction for any even history length.
\end{proof}

The construction works because both histories share one large oscillation, $\pm\sqrt{1-r^2}\,w$, which cancels in each temporal mean. What survives averaging is a short mean of length $r$, and normalization then presents the small difference between the two means as a full city angle. The message for screening is that city directions alone cannot guarantee novelty for arbitrary histories; \cref{thm:sharp} shows what retaining the lengths recovers.

\begin{example}[A finite positive-IC construction]\label{ex:finite}
Taking $\phi=120^\circ$ and $r=0.1$ in \cref{eq:counterexample} gives $D=0.985$ and IC $0.05$ at each date. Widely separated cities can therefore arise from two almost identical histories with positive realized target alignment.
\end{example}

The converse also fails: identical cities do not imply similar histories. Fix $a\in(0,1)$ and orthonormal $e,u,v\in H$. On two dates let
\[
 X_\sigma=ae+\sigma\sqrt{1-a^2}\,u,\qquad
 Y_\sigma=ae+\sigma\sqrt{1-a^2}\,v.
\]
Both cities are $e$, but $D=a^2$, which can be arbitrarily small. More generally, if a population has positive mean IC and zero mean transverse component for every signal, all population cities coincide at the North Pole. A first-moment summary then loses differences carried by temporal fluctuations.

\section{Sharp deterministic certificates}\label{sec:sharp}

The retained lengths quantify the energy discarded by temporal averaging. Define
\begin{equation}\label{eq:residuals}
 A_{i,t}=X_{i,t}-m_i,\qquad
 K_{ij}=\frac1T\sum_t\ip{A_{i,t}}{A_{j,t}}.
\end{equation}
These residuals are deviations from the signal's own time average. They should not be confused with the components orthogonal to the realized-return direction studied in~\cite{nunesIC}. Write
\[
 B_{ij}=\sqrt{(1-r_i^2)(1-r_j^2)}
\]
for the residual-energy radius.

Suppose the platform also publishes each retained length $r_i$. How much can then be said about $D_{ij}$ for an arbitrary pair? The next theorem gives an interval that always holds and that, under the stated dimension conditions, cannot be narrowed using cities and lengths alone.

\Needspace{14\baselineskip}
\begin{theorem}[Sharp bounds from cities and retained lengths]\label{thm:sharp}
For any two unit-vector histories with nonzero means,
\begin{equation}\label{eq:decomposition}
 D_{ij}=r_ir_jC_{ij}+K_{ij},
\end{equation}
and
\begin{equation}\label{eq:sharp}
 \boxed{\left|D_{ij}-r_ir_jC_{ij}\right|
 \le B_{ij}.}
\end{equation}
For $q\ge3$ and every even $T\ge2$, the entire interval in \cref{eq:sharp} is attainable for arbitrary prescribed city directions and retained lengths $r_i,r_j\in(0,1]$. For $q\ge4$, this statement holds for every $T\ge2$, including odd lengths.
\end{theorem}
\begin{proof}
Expand $X_{i,t}=m_i+A_{i,t}$ and use $\sum_tA_{i,t}=0$ to obtain \cref{eq:decomposition}. The unit-vector constraint gives
\begin{equation}\label{eq:energy}
 \frac1T\sum_t\norm{A_{i,t}}^2=1-r_i^2.
\end{equation}
Cauchy--Schwarz in the history Hilbert space yields \cref{eq:sharp}.

For attainability, fix any $z\in[-1,1]$. Since $q\ge3$, choose a unit vector $w$ orthogonal to both prescribed cities. Choose a unit vector $v$ orthogonal to $c_j$ and $w$, and put
\[
 a=w,\qquad b=zw+\sqrt{1-z^2}\,v.
\]
Then $a\perp c_i$, $b\perp c_j$, $\norm a=\norm b=1$, and $\ip ab=z$. Each residual needs to be orthogonal only to its own city. For two dates indexed by $\sigma\in\{-1,1\}$, define
\[
 X_{i,\sigma}=r_ic_i+\sigma\sqrt{1-r_i^2}\,a,\qquad
 X_{j,\sigma}=r_jc_j+\sigma\sqrt{1-r_j^2}\,b.
\]
Each observation is unit, the means are $r_ic_i$ and $r_jc_j$, and averaging gives $K_{ij}=zB_{ij}$. Varying $z$ attains the full interval already at $T=2$. Repeating the pair gives every even $T$. If either retained length is one, the interval is a singleton and the same construction applies.

For $q\ge4$, choose an orthonormal pair $w,v$ perpendicular to both cities, and put
\[
 u_t=\cos(2\pi t/T)w+\sin(2\pi t/T)v,\qquad 1\le t\le T.
\]
These unit vectors sum to zero for every $T\ge2$. Let $R_\beta$ rotate their plane by $\beta$ and take $\cos\beta=z$. The histories
\[
 X_{i,t}=r_ic_i+\sqrt{1-r_i^2}\,u_t,\qquad
 X_{j,t}=r_jc_j+\sqrt{1-r_j^2}\,R_\beta u_t
\]
are unit, have the prescribed means, and satisfy $K_{ij}=B_{ij}\cos\beta=zB_{ij}$. This proves the extension to odd lengths.
\end{proof}

In words, the full-history correlation splits into a part the compressed summary sees, $r_ir_jC_{ij}$, and a part it does not, $K_{ij}$, whose size is limited by how much of each history was averaged away. Long retained lengths leave a narrow interval and a usable certificate; short ones leave an interval close to $[-1,1]$. Sharpness means that this width is a limit on the information in the summary, not slack in the proof. \Cref{sec:positiveic} shows the consequence for the competition sample, where retained lengths are small and no pair is certified at threshold $0.5$.

The dimension threshold is optimal for a guarantee covering every even history length. For $q=2$, $T=2$, and $0<r_i,r_j<1$, the two observations of each history must be its mean plus or minus a residual orthogonal to that mean. If $J$ is a quarter-turn in the plane, their residual directions can only be $\pm Jc_i$ and $\pm Jc_j$. Consequently
\[
 K_{ij}\in\left\{-B_{ij}C_{ij},\,B_{ij}C_{ij}\right\}.
\]
For noncollinear cities, $|C_{ij}|<1$, so neither endpoint $\pm B_{ij}$ is attainable. Coincident or antipodal cities attain the endpoints, but still only two values at $T=2$. Thus dimension two cannot support the uniform full-interval statement.

Sharpness is proved only under the stated conditions on dimension and history length. We do not establish full-interval attainability for odd $T$ in dimension $q=3$; the two-date argument says nothing about that case. The inequality itself holds for every $T$ and every dimension. If a mean is zero, the inequality remains valid with $\ip{m_i}{m_j}$ in place of $r_ir_jC_{ij}$, and no city needs to be defined.

The interval turns directly into a screening rule: a pair is certified below a correlation threshold $\tau$ whenever the upper end of its interval is. Two versions are needed, because a one-sided cap treats a nearly opposite signal as novel and an absolute cap does not.

\begin{corollary}[Signed and absolute novelty certificates]\label{cor:screen}
Let $\tau\in[0,1]$ be a correlation threshold. Then the following implications hold deterministically:
\begin{align}
 r_ir_jC_{ij}+B_{ij}\le\tau
 &\quad\Longrightarrow\quad D_{ij}\le\tau,\label{eq:signedscreen}\\
 r_ir_j|C_{ij}|+B_{ij}\le\tau
 &\quad\Longrightarrow\quad |D_{ij}|\le\tau.\label{eq:absscreen}
\end{align}
\end{corollary}
\begin{proof}
Apply the upper endpoint in \cref{eq:sharp} for the first implication and the triangle inequality for the second.
\end{proof}

Both rules are sufficient conditions computed from cities and retained lengths only. A pair that fails them has not been shown to be a duplicate; the summary is inconclusive.

For example, orthogonal cities with $r_i=r_j=0.9$ certify $|D_{ij}|\le0.19$. When both lengths are small, the residual radius is close to one and the same bound may be inconclusive.

\begin{corollary}[A distance guarantee]\label{cor:distance}
For nonzero means,
\begin{equation}\label{eq:distance}
 2(1-D_{ij})\ge (r_i-r_j)^2+2r_ir_j(1-C_{ij}).
\end{equation}
If $r_i,r_j\ge r_0>0$, then
\begin{equation}\label{eq:sine}
 \sin\frac{\theta_{ij}}2\ge r_0\sin\frac{\phi_{ij}}2.
\end{equation}
\end{corollary}
\begin{proof}
Jensen's inequality gives
\begin{align*}
 2(1-D_{ij})
 &=\frac1T\sum_t\norm{X_{i,t}-X_{j,t}}^2\\
 &\ge\norm{m_i-m_j}^2
 =(r_i-r_j)^2+2r_ir_j(1-C_{ij}).
\end{align*}
Drop the squared radial difference and use $1-\cos\alpha=2\sin^2(\alpha/2)$.
\end{proof}

\Cref{eq:sine} reads as a discount: the half-angle sine of the history separation is at least that of the city separation, scaled by the retained-length floor $r_0$. With $r_0$ near one, city separation carries over almost in full; with small $r_0$ the guarantee is weak. \Cref{sec:stability} applies it with $r_0$ equal to a persistent IC margin, giving \cref{eq:predictivenovelty}.

When $r_i=r_j=r$, the upper endpoint simplifies to $D_{ij}\le1-r^2(1-C_{ij})$. \Cref{tab:thresholds} below translates this expression into the commonly used $0.5$ one-sided threshold.

\begin{table}[H]
\centering\small
\caption{Equal retained lengths: a deterministic certificate for the one-sided rule $D\le0.5$. These thresholds are not certificates for the absolute-correlation rule.}
\label{tab:thresholds}
\begin{tabular}{@{}ccc@{}}
\toprule
City angle $\phi$ & City correlation $C$ & Sufficient retained length $r$\\
\midrule
$60^\circ$ & $0.5$ & $r=1$\\
$90^\circ$ & $0$ & $r\ge1/\sqrt2\simeq0.7071$\\
$120^\circ$ & $-0.5$ & $r\ge1/\sqrt3\simeq0.5774$\\
\bottomrule
\end{tabular}
\end{table}

\subsection{When does compression understate angular separation?}

There is no universal ordering between $D_{ij}$ and $C_{ij}$. Their difference is exactly
\begin{equation}\label{eq:ordering}
 D_{ij}-C_{ij}=K_{ij}-(1-r_ir_j)C_{ij}.
\end{equation}
\Needspace{13\baselineskip}
A user of the published cities usually wants something simpler than an interval: is the city angle a conservative estimate, so that the histories are at least as far apart as their cities? The identity above answers this exactly.

\begin{corollary}[Signal separation at least city separation]\label{cor:angleordering}
For any two unit-vector histories with nonzero means,
\[
 \theta_{ij}\ge\phi_{ij}
 \quad\Longleftrightarrow\quad
 K_{ij}\le(1-r_ir_j)C_{ij}.
\]
In particular, if the city angle satisfies $0\le\phi_{ij}\le\pi/2$ and the temporal residual cross-moment satisfies $K_{ij}\le0$, then
\[
 \boxed{\theta_{ij}\ge\phi_{ij}.}
\]
Thus the angle between the full signal histories is at least the angle between their cities.
\end{corollary}
\begin{proof}
Since $\arccos$ is decreasing on $[-1,1]$, the inequality $\theta_{ij}\ge\phi_{ij}$ is equivalent to $D_{ij}\le C_{ij}$. \Cref{eq:ordering} gives the stated criterion. For the sufficient conditions, $0<r_i,r_j\le1$ and $C_{ij}=\cos\phi_{ij}\ge0$ imply $(1-r_ir_j)C_{ij}\ge0$, so $K_{ij}\le0$ satisfies that criterion.
\end{proof}

The criterion says that the city angle is conservative exactly when the residual co-movement $K_{ij}$, which the cities never saw, does not exceed a margin fixed by the cities and retained lengths. It involves no probability model and can be checked on any pair whose histories are available. \Cref{sec:empirical} evaluates it for every pair in the competition sample; \cref{sec:prob} asks how often it holds under specified sampling laws.

The angle restriction matters: if $K_{ij}=0$, $r_ir_j<1$, and $\phi_{ij}>\pi/2$, then $D_{ij}=r_ir_jC_{ij}>C_{ij}$ and hence $\theta_{ij}<\phi_{ij}$. Positive IC alone also does not guarantee the ordering, as \cref{thm:impossible} shows. If one ordering predominates in data, that does not make it a universal inequality; additional assumptions are needed.

\section{Probabilistic separation}\label{sec:prob}

The deterministic bound allows residual fluctuations to align in the worst possible way. Improving on it probabilistically requires an explicit sampling law, and this section considers five. We begin with histories drawn uniformly and selected for positive average IC, and compute the probability that a pair preserves separation (\cref{sec:band,sec:uniformnorth}). We then allow dependence across dates within each history, keeping the two histories independent and rotationally symmetric (\cref{sec:dependentnull}). The third model holds the two observed collections fixed and randomizes only how their dates are matched. The fourth assumes independent dates but allows the two signals to depend on each other at each date. The fifth, a stationary ergodic model, permits both temporal and cross-signal dependence and identifies a population margin.

\subsection{The probability mass beyond nonpositive residual agreement}\label{sec:band}

The condition $K_{ij}\le0$ in \cref{cor:angleordering} is sufficient but not necessary, so it misses some pairs that preserve separation. For any random pair of histories with nonzero means, suppress the indices and set
\begin{equation}\label{eq:separationmargin}
 J=\{C\ge0\},\qquad \varepsilon=(1-r_Xr_Y)C.
\end{equation}
On $J$, the margin $\varepsilon$ is nonnegative. Positive residual agreement is allowed, provided it does not exceed this margin.

\begin{proposition}[Exact decomposition of separation probability]\label{prop:band}
For any probability law on pairs of histories with nonzero means and $\Prob(J)>0$,
\begin{equation}\label{eq:bandprobability}
 \Prob(\theta\ge\phi\mid J)
 =\Prob(K\le0\mid J)
 +\Prob(0<K\le\varepsilon\mid J).
\end{equation}
\end{proposition}
\begin{proof}
By \cref{eq:ordering}, $\theta\ge\phi$ if and only if $K\le\varepsilon$. On $J$, the latter event is the disjoint union of $\{K\le0\}$ and $\{0<K\le\varepsilon\}$.
\end{proof}

The second term is the probability of an additional band of pairs: those whose residual agreement is positive but no larger than the margin. The width of this band varies with the pair, so it is not, in general, an annulus of fixed angular width. We do not assume that $K$ and $\varepsilon$ are independent. If the conditional law of $K$ is symmetric and has no atom at zero, the first term is exactly one-half; that symmetry must be shown under the same conditioning as the conclusion it supports. The next model gives an asymptotic version of this statement and a stronger limit for the whole admissible region.

\subsection{A uniform continuum of northern histories}\label{sec:uniformnorth}

Let $\sigma_q$ be normalized surface measure on $\sphere^{q-1}$, and let $\mu_T=\sigma_q^{\otimes T}$ be product measure on histories. Define
\begin{equation}\label{eq:northernhistoryspace}
 \Omega_T^+=\left\{(X_1,\ldots,X_T):
 \ip{T^{-1}\textstyle\sum_tX_t}{e}>0\right\}.
\end{equation}
This set has $\mu_T$-measure one-half. Two independent draws from $\mu_T(\,\cdot\mid\Omega_T^+)$ describe a uniform pair of histories whose temporal means lie above the equator. Equivalently, draw all $2T$ observations independently and uniformly, then condition on
\begin{equation}\label{eq:northernselection}
 H_T=\{\ip{m_X}{e}>0,\ \ip{m_Y}{e}>0\},
 \qquad \Prob(H_T)=\tfrac14.
\end{equation}
Write $\Prob_T^+$ and $\E_T^+$ for this conditional pair law and expectation, and put $J_T=\{C_T\ge0\}=\{\phi_T\le\pi/2\}$. Thus the superscript $+$ selects two histories with positive average IC, while $J_T$ additionally restricts their city angle to at most $90^\circ$. The means are nonzero almost surely. We refer to this independent uniform-history model as the \emph{uniform null}, and to its restriction $\Prob_T^+$ to histories with positive average IC as the \emph{northern null}; a \emph{northern} history or pair is one selected in this way. Conditioning preserves independence between the two histories and exchangeability within each one; observations within a selected history are no longer independent.

This model is a reference point, not a description of a real library. Every direction is equally likely at every date, the two signals are unrelated, and the only selection is that each ends with positive average IC. The next theorem records how large the unseen residual $K_T$ is in that world, and shows that, in long histories, selection on IC does not tilt its sign.

\begin{theorem}[Uniform northern histories: exact moments and balanced residual signs]\label{thm:northernnull}
Fix $q\ge2$ and let $T\ge2$. Under the law $\Prob_T^+$,
\begin{equation}\label{eq:northernmoments}
 \E_T^+K_T=0,\qquad
 \Var_T^+(K_T)=\frac{T-1}{qT^2},\qquad
 \E_T^+r_X^2=\E_T^+r_Y^2=\frac1T.
\end{equation}
As $T\to\infty$ with $q$ fixed,
\begin{equation}\label{eq:northernjointclt}
 \bigl(\sqrt{qT}K_T,\sqrt{qT}m_X,\sqrt{qT}m_Y\bigr)
 \dto (Z,G_+,H_+),
\end{equation}
where $Z\sim N(0,1)$ and $G_+,H_+$ are standard $q$-dimensional Gaussians conditioned to have positive projection on $e$, all three independent. Consequently,
\begin{equation}\label{eq:northernsign}
 \Prob_T^+(K_T\le0)\longrightarrow\tfrac12,
 \qquad
 \Prob_T^+(K_T\le0\mid J_T)\longrightarrow\tfrac12.
\end{equation}
\end{theorem}
\begin{proof}
Exchangeability gives $\E_T^+(X_t-m_X)=0$ for every $t$. Independence of the two selected histories then gives $\E_T^+K_T=0$.

Before conditioning, $\E r_X^2=\E r_Y^2=1/T$. Put $a_t=X_t-m_X$; since $\sum_ta_t=0$, independence and $\E[Y_tY_t^\top]=\id_q/q$ give
\[
 \E[K_T^2\mid X_1,\ldots,X_T]
 =\frac{1}{qT^2}\sum_t\norm{a_t}^2
 =\frac{1-r_X^2}{qT}.
\]
Thus $\E K_T^2=(T-1)/(qT^2)$. Reversing the sign of either entire history preserves product measure and $K_T^2$, while changing the corresponding mean's hemisphere. These reversals interchange the four equally likely sign combinations in \cref{eq:northernselection}, so conditioning on $H_T$ leaves the second moment unchanged. The same argument applies to each $r^2$. Together with the conditional zero mean, this proves \cref{eq:northernmoments}.

Under the unconditioned law the independent and identically distributed (iid) vectors
\[
 W_t=\sqrt q\,\bigl(\ip{X_t}{Y_t},X_t,Y_t\bigr)\in\R^{1+2q}
\]
have mean zero and covariance $\id_{1+2q}$. The multivariate central limit theorem~\cite[Section~3.10]{durrett} therefore gives independent standard Gaussian limits for $\sqrt{qT}D_T$, $\sqrt{qT}m_X$, and $\sqrt{qT}m_Y$. Also
\[
 \sqrt{qT}\,\ip{m_X}{m_Y}\longrightarrow0
 \quad\hbox{in probability},
\]
so $D_T$ can be replaced by $K_T$. The limiting event that both Gaussian means have positive projection on $e$ has probability $1/4$ and boundary of probability zero. Conditioning on $H_T$ therefore gives \cref{eq:northernjointclt}. The event $G_+^\top H_+\ge0$ also has positive probability and boundary of probability zero. It depends only on the Gaussian means, which are independent of $Z$. This proves both sign limits in \cref{eq:northernsign}.
\end{proof}

Two facts carry forward. First, both the residual $K_T$ and the retained lengths are small, of order $1/\sqrt{qT}$ and $1/\sqrt T$ respectively: unrelated random histories are nearly orthogonal and their means nearly vanish. Second, the residual is asymptotically as likely to be negative as positive, with or without the restriction to city angles of at most $90^\circ$.

The practical question for this null model is how often the city angle is conservative, that is, how often $\theta_T\ge\phi_T$. The answer turns almost entirely on whether the city angle is acute or obtuse.

\begin{corollary}[Positive-IC separation preservation for city angles at most $90^\circ$]\label{cor:northernseparation}
Under the independent uniform-history model of \cref{thm:northernnull}, select both histories for strictly positive average IC. Set $\varepsilon_T=(1-r_Xr_Y)C_T$, $E_T=\{\theta_T\ge\phi_T\}$, and $J_T=\{C_T\ge0\}=\{\phi_T\le\pi/2\}$. Among these positive-IC pairs whose city angle is at most $90^\circ$, the probability that history separation is at least city separation tends to one as $T\to\infty$. More precisely, $E_T$ and $J_T$ asymptotically coincide:
\begin{equation}\label{eq:northerneventequivalence}
 \Prob_T^+(E_T\mathbin{\triangle}J_T)\longrightarrow0,
\end{equation}
where $\triangle$ denotes symmetric difference. In particular,
\begin{equation}\label{eq:northernseparation}
 \Prob_T^+(\theta_T\ge\phi_T\mid \phi_T\le\pi/2)\longrightarrow1,
 \qquad
 \Prob_T^+(\theta_T\ge\phi_T\mid \phi_T>\pi/2)\longrightarrow0.
\end{equation}
The admissible-band identity also gives
\begin{equation}\label{eq:northernbandlimit}
 \Prob_T^+(0<K_T\le\varepsilon_T\mid J_T)\longrightarrow\tfrac12.
\end{equation}
Over all northern pairs, without conditioning on $J_T$, the limit is instead
\begin{equation}\label{eq:northernallpairs}
 \Prob_T^+(\theta_T\ge\phi_T)\longrightarrow\gamma_q,
 \qquad
 \gamma_q=\Prob(U^\top V\ge0),
\end{equation}
where $U,V$ are independent uniform directions in the open northern hemisphere of $\sphere^{q-1}$. Explicitly,
\begin{equation}\label{eq:northernanglemeasure}
 \gamma_q=
 \frac{2\Gamma(q/2)}{\sqrt\pi\,\Gamma((q-1)/2)}
 \int_0^{\pi/2}\left(1-\frac{u}{\pi}\right)\sin^{q-2}u\,du.
\end{equation}
For example, $\gamma_2=3/4$, $\gamma_3=1-1/\pi$, and $\gamma_{19}\simeq0.560117$. Moreover,
\begin{equation}\label{eq:gammaasymptotic}
 \gamma_q=\frac12+\frac{\sqrt2}{\pi^{3/2}\sqrt q}
 +O(q^{-3/2})\qquad(q\to\infty).
\end{equation}
The coefficient of $q^{-1/2}$ is approximately $0.254$.
\end{corollary}
\begin{proof}
Rotational invariance makes each unconditioned mean direction uniform. After northern selection, $c_X,c_Y$ are independent uniform northern directions for every $T$. Consequently $C_T$ has a distribution independent of $T$, with no atom at zero, and $\Prob_T^+(J_T)=\gamma_q>0$.

Under the unconditioned law, $D_T\to0$ in probability by the law of large numbers. The same holds under $\Prob_T^+$, since $\Prob(H_T)=1/4$. For any $\delta>0$, the events $\{D_T\le C_T\}$ and $J_T$ can disagree only if $|D_T|>\delta$ or $|C_T|\le\delta$. Letting first $T\to\infty$ and then $\delta\downarrow0$ proves \cref{eq:northerneventequivalence}. Since both $J_T$ and its complement have fixed positive probability, both conditional limits in \cref{eq:northernseparation} and the unconditional limit \cref{eq:northernallpairs} follow. Subtracting \cref{eq:northernsign} from \cref{eq:bandprobability} gives \cref{eq:northernbandlimit}.

For \cref{eq:northernanglemeasure}, the angle $u$ between two independent full-sphere directions has density~\cite{caifanjiang}
\[
 f_q(u)=\frac{\Gamma(q/2)}{\sqrt\pi\,\Gamma((q-1)/2)}\sin^{q-2}u,
 \qquad 0<u<\pi.
\]
For a uniformly oriented pair at angle $u$, the probability that both directions lie north is $(\pi-u)/(2\pi)$: the allowed directions of the pole's projection into their two-dimensional span occupy an arc of length $\pi-u$. Since the unconditional probability of both directions being north is $1/4$, their angle density after northern selection is $2(1-u/\pi)f_q(u)$. Integrating over $0<u<\pi/2$ proves the formula and its stated special cases.

If $Z_q$ is the first coordinate of a uniform point on $\sphere^{q-1}$, the same formula gives
\[
 \gamma_q=\frac12+\frac1\pi\E\arcsin|Z_q|,
 \qquad
 \E|Z_q|^a=
 \frac{\Gamma(q/2)\Gamma((a+1)/2)}{\sqrt\pi\,\Gamma((q+a)/2)}
 \quad(a>-1).
\]
The ratio $(\arcsin z-z)/z^3$ is bounded on $0<z\le1$. The displayed moments with $a=1,3$, together with Stirling's formula for the gamma ratios, therefore give \cref{eq:gammaasymptotic}.
\end{proof}

The mechanism is near-orthogonality of independent random histories: $D_T\to0$, so their angle approaches $90^\circ$. An acute city angle then understates their separation, while an obtuse city angle overstates it. The band decomposition accounts for the successful acute pairs; it is not a separate cause of this ordering. The conditional limits $1$ and $0$ also hold without northern selection, by the same argument. Selection does not cause the conditional limit of one. Its effect is to change the unconditional limit from $1/2$ to $\gamma_q$, about $56\%$ for $q=19$ (20 assets). For any fixed $0<\eta<1$, the proof also gives $\Prob_T^+(\theta_T\ge\phi_T\mid C_T\ge\eta)\to1$.

\begin{remark}[Finite histories and the scope of the null model]\label{rem:northernnullscope}
For $T=2$, swapping the two dates of one history preserves its mean, $H_T$, and $J_T$, while reversing $K_T$. Since $K_T$ has no atom at zero, both sign probabilities in \cref{eq:northernsign} are then exactly one-half. In the special case $q=T=2$, the ordering itself is exact almost surely: $E_T=J_T$. Indeed, \cref{sec:sharp} gives $K=\pm BC$ and hence $D=(r_Xr_Y\pm B)C$. Cauchy--Schwarz gives $B\le1-r_Xr_Y$, and $r_Xr_Y+B=1$ only when $r_X=r_Y$. This equality has probability zero for independent continuous histories, so the coefficient of $C$ in $D$ is strictly less than one almost surely. This yields both conditional ordering probabilities $1$ and $0$ already at $q=T=2$.

For general $T$, reflecting residuals about a fixed mean would send $X_t$ to $2m_X-X_t$, which need not have unit norm. Central symmetry of the residual law therefore cannot simply be assumed on a product of spheres.

Northern selection means positive empirical average IC, with no fixed positive margin and no requirement of positive IC at each date. It selects half of a zero-mean population and is not a model of persistent forecasting skill. The conditional law is specified separately at each $T$; it does not condition an infinite history to remain north at every prefix. Shrinking retained lengths explain this particular null model but are not necessary for separation preservation. \Cref{thm:populationordering} gives the appropriate margin condition when population means are nonzero. All these angular statements concern signed correlation, not an absolute-correlation novelty threshold.
\end{remark}

The limits above do not say how long a history must be. The next theorem replaces them by an explicit error bound that holds at every history length $T$ and does not worsen as the number of assets grows.

\begin{theorem}[A finite-sample ordering bound uniform in dimension]\label{thm:uniformorderingrate}
For every $q\ge2$ and $T\ge2$ under the uniform northern-history law, which selects both histories for positive average IC, let $E_T=\{\theta_T\ge\phi_T\}$ and $J_T=\{C_T\ge0\}=\{\phi_T\le\pi/2\}$. Then
\begin{equation}\label{eq:uniformorderingrate}
 \Prob_T^+(E_T\mathbin{\triangle}J_T)
 \le b_T:=\frac1T+\left(1-\frac1T\right)\sqrt{\frac{\pi}{2T}}.
\end{equation}
In particular,
\begin{align}
 \Prob_T^+(\theta_T<\phi_T\mid J_T)&\le\frac{b_T}{\gamma_q},\label{eq:uniformacuterate}\\
 \Prob_T^+(\theta_T\ge\phi_T\mid C_T<0)&\le\frac{b_T}{1-\gamma_q}.\label{eq:uniformobtuserate}
\end{align}
Equivalently, the preservation probability for positive-IC pairs with city angle at most $90^\circ$ obeys
\[
 \boxed{\Prob_T^+(\theta_T\ge\phi_T\mid\phi_T\le\pi/2)
 \ge \max\left\{0,\,1-\frac{b_T}{\gamma_q}\right\}.}
\]
Both errors are $O(T^{-1/2})$ uniformly over $q\ge2$. Thus the event-equivalence and conditional ordering limits remain valid for any dimension sequence $q=q_T\ge2$, and
\begin{equation}\label{eq:growingdimensionordering}
 \left|\Prob_T^+(E_T)-\gamma_{q_T}\right|\le b_T.
\end{equation}
If also $q_T\to\infty$, the unconditional ordering probability tends to one-half.
\end{theorem}
\begin{proof}
Condition first on the two mean vectors and put $a_t=X_t-m_X$, $b_t=Y_t-m_Y$. The two residual histories remain independent and each is exchangeable with sum zero. Define
\[
 M_X=\E\left[\frac1T\sum_ta_ta_t^\top\,\middle|\,m_X\right],
 \qquad
 M_Y=\E\left[\frac1T\sum_tb_tb_t^\top\,\middle|\,m_Y\right].
\]
Every conditional residual mean is zero. For distinct dates, exchangeability and the zero sum give $\E[a_ta_s^\top\mid m_X]=-M_X/(T-1)$, with the analogous identity for $Y$. Expanding $K_T^2$ therefore yields
\begin{equation}\label{eq:conditionalmeanvariance}
 \E[K_T\mid m_X,m_Y]=0,\qquad
 v:=\Var(K_T\mid m_X,m_Y)=\frac{\tr(M_XM_Y)}{T-1}.
\end{equation}
These conditional laws are unchanged by northern selection once the means are specified.

Put $\tau_X=1-r_X^2$, $\tau_Y=1-r_Y^2$, and $C=c_X^\top c_Y$. Rotational invariance about $c_X$ implies that $M_X/\tau_X$ is a convex combination of
\[
 P_X=c_Xc_X^\top,
 \qquad Q_X=\frac{\id_q-P_X}{q-1},
\]
and similarly for $Y$. The traces of the four possible products are
\[
 C^2,\qquad \frac{1-C^2}{q-1},\qquad
 \frac{1-C^2}{q-1},\qquad
 \frac{q-2+C^2}{(q-1)^2}.
\]
Each is at most $C^2+1/(q-1)$. Degenerate zero-residual cases satisfy the same bound directly. Consequently
\[
 v\le\frac{\tau_X\tau_Y}{T-1}\left(C^2+\frac1{q-1}\right).
\]
For $C>0$, the event $E_T\triangle J_T$ requires $K_T>(1-r_Xr_Y)C$; for $C<0$, it requires $K_T\le(1-r_Xr_Y)C$. Apply the appropriate one-sided Cantelli inequality~\cite[Exercise~1.6.4]{durrett} to \cref{eq:conditionalmeanvariance}. Since
\[
 \tau_X\tau_Y\le(1-r_Xr_Y)^2
\]
by Cauchy--Schwarz, the conditional error probability is bounded by
\begin{equation}\label{eq:conditionalorderingbound}
 \frac{v}{v+(1-r_Xr_Y)^2C^2}
 \le\frac{C^2+1/(q-1)}{TC^2+1/(q-1)}
 =\frac1T+\frac{1-1/T}{1+T(q-1)C^2}.
\end{equation}
The null event $C=0$ does not affect the integrated bound.

The density of $C$ after northern selection is
\[
 g_q(c)=\left(1+\frac2\pi\arcsin c\right)h_q(c),\qquad
 h_q(c)=\frac{\Gamma(q/2)}{\sqrt\pi\,\Gamma((q-1)/2)}(1-c^2)^{(q-3)/2}.
\]
Here $h_q$ is the full-sphere inner-product density on $(-1,1)$. Because its multiplier differs from one by an odd function, every even function of $C$ has the same expectation before and after northern selection. For $q\ge3$, $h_q(c)\le A_q:=\Gamma(q/2)/(\sqrt\pi\,\Gamma((q-1)/2))$, whence
\[
 \E_T^+\frac1{1+T(q-1)C^2}
 \le A_q\int_{-\infty}^{\infty}\frac{dc}{1+T(q-1)c^2}
 =\frac{\pi A_q}{\sqrt{T(q-1)}}
 \le\sqrt{\frac\pi{2T}}.
\]
The last step uses $\Gamma(x+1/2)^2\le\Gamma(x)\Gamma(x+1)=x\Gamma(x)^2$, which follows from Cauchy--Schwarz applied to the gamma integrals. For $q=2$, substituting $c=\cos u$ in the arcsine density gives the expectation exactly as $(1+T)^{-1/2}$, which obeys the same upper bound. Integrating \cref{eq:conditionalorderingbound} proves \cref{eq:uniformorderingrate}.

Finally, $1/2<\gamma_q\le3/4$ for all $q\ge2$. To see the upper bound, couple $|Z_q|=|G_1|/(\sum_{j=1}^qG_j^2)^{1/2}$ using iid standard Gaussians: it decreases with $q$, as does $\E\arcsin|Z_q|$. The lower bound follows from its strict positivity. Thus the denominators in \cref{eq:uniformacuterate,eq:uniformobtuserate} are at least $1/2$ and $1/4$, respectively. All ordering bounds are uniform in $q$, and \cref{eq:gammaasymptotic} gives the final assertion.
\end{proof}

The ordering can fail only when the small residual $K_T$ outweighs the margin $(1-r_Xr_Y)C_T$, which is likely only when the two cities are close to orthogonal; the bound $b_T$ controls the probability of that near-tie. Uniformity in $q$ means that the same statement covers a universe of 20 assets and one of 2,000.

This upper bound is rigorous, but its constants are conservative. It gives no matching lower rate and does not predict the observed error percentage. For example, at $T=100$ it gives $b_T\simeq0.1341$; with $q=19$, the acute-pair bound is about $0.2394$. The theorem controls the ordering error even when the dimension grows. By contrast, the joint Gaussian limit and the residual-sign statements in \cref{thm:northernnull} are proved for fixed dimension only; we do not assert them as central limit theorems in growing dimension.

\subsection{Temporal dependence within rotationally symmetric histories}\label{sec:dependentnull}

The next result keeps the two histories independent of each other but allows arbitrary dependence across dates within each one. It assumes \emph{pathwise rotational symmetry}: the law of a history is unchanged when one common orthogonal transformation is applied to all of its dates. This is a stronger requirement than uniform marginals at each date. The result enlarges the class of sampling laws covered by the ordering bound; for independent dates, \cref{thm:uniformorderingrate} remains sharper.

Real signal histories are serially dependent: forecasts persist and forward-return targets overlap. The next theorem shows how much of the ordering result survives.

\begin{theorem}[Temporal dependence under pathwise rotational symmetry]\label{thm:dependentordering}
Use $H_T$, $E_T$, $J_T$, and $\gamma_q$ as above, and let $q\ge2$. Suppose the entire histories $(X_1,\ldots,X_T)$ and $(Y_1,\ldots,Y_T)$ are independent, have unit rows, and have nonzero empirical means almost surely. Suppose each history's law is invariant under applying any one orthogonal transformation to all of its dates. Define their expected temporal Gram matrices by
\[
 R_X(s,t)=\E(X_s^\top X_t),\qquad R_Y(s,t)=\E(Y_s^\top Y_t),
\]
and set
\[
 a_T=\frac{\tr(R_XR_Y)}{T^2},\qquad
 \mathcal B(a)=\min\left\{1,\ a+\sqrt{\frac8\pi}\sqrt{a(1-a)}\right\}.
\]
Then $0\le a_T\le1$, $\Prob_T^+(J_T)=\gamma_q$, and
\begin{equation}\label{eq:dependenterror}
 \Prob_T^+(E_T\mathbin{\triangle}J_T)\le \mathcal B(a_T).
\end{equation}
Consequently,
\begin{align}
 \Prob_T^+(E_T\mid J_T)&\ge\max\{0,1-\mathcal B(a_T)/\gamma_q\},\label{eq:dependentacute}\\
 \Prob_T^+(E_T\mid J_T^c)&\le\min\{1,\mathcal B(a_T)/(1-\gamma_q)\},\\
 |\Prob_T^+(E_T)-\gamma_q|&\le \mathcal B(a_T).
\end{align}
In particular, for fixed $q$, if $a_T\to0$, the respective probabilities converge to $1$, $0$, and $\gamma_q$. No independence, exchangeability, stationarity, or mixing across dates is required for the finite-sample bound.
\end{theorem}

Three points about the scope of the theorem deserve emphasis. Its assumptions concern the complete law of a history: uniform marginal directions at each date do not, alone, imply pathwise rotational symmetry. Independence of the two entire histories remains a substantive assumption. And the theorem does not extend the residual-sign central limit theorem in \cref{thm:northernnull}.

\begin{proof}
Rotational symmetry makes each mean direction uniform and independent of its length. After northern selection the directions remain independent and are uniform on the northern hemisphere, so $\Prob_T^+(J_T)=\gamma_q$. Conditional matrices are taken in rotationally equivariant versions.

Condition on $m_X$ and put $P_X=c_Xc_X^\top$. Invariance under the orthogonal transformations fixing $c_X$ gives
\[
 \E[X_sX_t^\top\mid m_X]
 =A_X(s,t)P_X+\frac{B_X(s,t)}{q-1}(\id_q-P_X),
\]
where the $T\times T$ matrices $A_X$ and $B_X$ are the conditional Gram matrices of the components parallel and perpendicular to $c_X$. Both are positive semidefinite. Their sum
\[
 F_X(s,t)=\E[X_s^\top X_t\mid m_X]
\]
has trace $T$. These scalar temporal matrices depend only on $r_X$, by rotational symmetry. Define $A_Y,B_Y,F_Y$ analogously. Conditional independence of the two histories yields, with $C=c_X^\top c_Y$,
\begin{align*}
 T^2\E[D_T^2\mid m_X,m_Y]
 ={}&C^2\tr(A_XA_Y)\\
 &+\frac{1-C^2}{q-1}\{\tr(A_XB_Y)+\tr(B_XA_Y)\}\\
 &+\frac{q-2+C^2}{(q-1)^2}\tr(B_XB_Y).
\end{align*}
Every trace product is nonnegative, and each displayed coefficient is at most $C^2+1/(q-1)$. Therefore, setting
\[
 h=\frac{\tr(F_XF_Y)}{T^2}\in[0,1],
\]
we obtain
\begin{equation}\label{eq:dependentsecond}
 \E[D_T^2\mid m_X,m_Y]\le h\left(C^2+\frac1{q-1}\right).
\end{equation}
The mismatch event $E_T\triangle J_T$ implies $|D_T|\ge|C|$. The event $C=0$ is null. Markov's inequality and \cref{eq:dependentsecond} give, for $0<h<1$,
\begin{align*}
 \Prob(E_T\triangle J_T\mid m_X,m_Y)
 &\le\min\left\{1,h+\frac{h}{(q-1)C^2}\right\}\\
 &=h+(1-h)\min\left\{1,\frac{h}{(1-h)(q-1)C^2}\right\}.
\end{align*}
These conditional laws are unchanged by northern selection once the means are specified. Also $h$ depends only on the two radii, and is independent of $C$ under the selected law.

The selected density of $C$ is
\[
 g_q(c)=\left(1+\frac2\pi\arcsin c\right)f_q(c),\qquad
 f_q(c)=\frac{\Gamma(q/2)}{\sqrt\pi\,\Gamma((q-1)/2)}(1-c^2)^{(q-3)/2}.
\]
Thus every even function of $C$ has the same expectation before and after northern selection. For $q\ge3$, let $A_q=\Gamma(q/2)/(\sqrt\pi\,\Gamma((q-1)/2))$. Since $f_q\le A_q$,
\[
 \E_T^+\min\{1,u^2/C^2\}
 \le A_q\int_{\mathbb R}\min\{1,u^2/c^2\}\,dc=4A_qu.
\]
Log-convexity of $\Gamma$ gives $A_q\le\sqrt{(q-1)/(2\pi)}$. Substituting $u^2=h/((1-h)(q-1))$ bounds the angularly integrated error by
\[
 h+\sqrt{8/\pi}\sqrt{h(1-h)}.
\]
For $q=2$, the full-sphere density is the arcsine density and, for $0\le u\le1$,
\[
 \E\min\{1,u^2/C^2\}
 =\frac2\pi\{\arcsin u+u\sqrt{1-u^2}\}\le\frac4\pi u.
\]
For $u>1$ the same inequality holds because the expectation is one. Since $4/\pi\le\sqrt{8/\pi}$, the stated bound also holds for $q=2$. The cases $h=0,1$ follow directly or by continuity.

Finally, northern selection does not affect the radii or $F_X,F_Y$. Independence and conditional expectation imply
\[
 \E_T^+h=\frac{\tr((\E F_X)(\E F_Y))}{T^2}
 =\frac{\tr(R_XR_Y)}{T^2}=a_T.
\]
Concavity of $h\mapsto\sqrt{h(1-h)}$ and Jensen's inequality prove \cref{eq:dependenterror}, also using the trivial bound one. The remaining assertions follow by dividing the appropriate mismatch probabilities by $\gamma_q$ or $1-\gamma_q$.
\end{proof}

The quantity $a_T$ summarizes how much the two histories repeat themselves over time: it equals $1/T$ when dates are independent and one when neither history ever changes. The theorem shows that the ordering conclusion needs $a_T$ to be small, not independence across dates. \Cref{cor:dependentmoments} expresses $a_T$ through lag correlations.

\begin{corollary}[Second moments and the temporal dependence factor]\label{cor:dependentmoments}
Under \cref{thm:dependentordering},
\begin{equation}\label{eq:dependentmoments}
 \E D_T=0,\qquad \E D_T^2=\E_T^+D_T^2=\frac{a_T}{q}\le\frac1q.
\end{equation}
The zero-mean identity refers to the unconditioned law. If the histories are stationary, put $\alpha_X(h)=\E\ip{X_0}{X_h}$ and $\alpha_Y(h)=\E\ip{Y_0}{Y_h}$. Then
\begin{equation}\label{eq:dependentlags}
 a_T=\frac{L_T}{T},\qquad
 L_T=1+2\sum_{h=1}^{T-1}\left(1-\frac hT\right)\alpha_X(h)\alpha_Y(h).
\end{equation}
If $\sum_{h\ge1}|\alpha_X(h)\alpha_Y(h)|<\infty$, then $L_T$ is bounded, the mismatch bound is $O(T^{-1/2})$, and the conditional preservation limits remain $1$ on $J_T$ and $0$ on its complement.
\end{corollary}
\begin{proof}
Pathwise rotational symmetry gives zero one-date means and $\E[X_sX_t^\top]=R_X(s,t)\id_q/q$. Independence of the two histories gives
\[
 \E[(X_s^\top Y_s)(X_t^\top Y_t)]=R_X(s,t)R_Y(s,t)/q.
\]
Summing over dates proves the unconditioned moment identities. Reversing either entire history leaves $D_T^2$ unchanged and interchanges the four equally likely signs of the two mean ICs, so its second moment is unchanged by northern selection. Stationarity groups the temporal Gram sum by lag, proving \cref{eq:dependentlags}. Absolute summability bounds $L_T$, and \cref{thm:dependentordering} gives the rate and limits.
\end{proof}

Whatever the temporal dependence, this null model caps $q\,\E D_T^2$ at one. An observed library whose average $qD_{ij}^2$ exceeds one therefore cannot be matched by adjusting serial dependence within the model; \cref{eq:empiricalmomentdiagnostic} applies this descriptive check to the competition sample.

Northern selection need not preserve $\E D_T=0$: two constant histories with independent uniform directions become independent uniform northern directions after selection, giving a positive mean inner product. The second-moment identity in \cref{eq:dependentmoments} is invariant under these sign selections, including any fixed choice of northern or southern hemispheres, but need not survive further conditioning on city angle or competition-quality criteria.

The dependence adjustment in \cref{eq:dependentlags} involves the \emph{product} of the two histories' lag correlations. It does not amount to substituting an effective sample size into \cref{thm:uniformorderingrate}. The two proofs differ: the dependent-history proof uses a second moment about zero and Markov's inequality, whereas the independent-date proof uses exchangeability and a centered residual to apply Cantelli's inequality. For that reason the dependent-history bound is looser even when $L_T=1$. It extends neither the residual-sign central limit theorem nor the independent-date concentration bound below.

\paragraph{Overlapping targets and repeated observations.}
If $T=\ell n$ and each of $n$ independent uniform vector observations is repeated $\ell$ times in aligned, equal blocks in both histories, then $m_X,m_Y,D_T,C_T$ are exactly those of the $n$-observation histories. Thus \cref{thm:uniformorderingrate} applies with $n=T/\ell$ exactly. Here $a_T=1/n$; the shorthand factor $Ta_T$ equals $\ell$, even though blocks anchored at a fixed origin need not define a stationary process.

A five-hour forward-return target does not imply this repetition model or independence of every fifth prediction. For illustration, scalar sums of five independent hourly innovations have lag correlations $4/5,3/5,2/5,1/5$ and zero thereafter. Their own long-run mean-variance factor is $1+2\sum_{h=1}^4(5-h)/5=5$. For the product of two independent such streams, it is $1+2\sum_{h=1}^4((5-h)/5)^2=3.4$. These are illustrative scalar calculations, not claims about the normalized signal vectors. Predictions and returns may have dependence beyond the mechanical overlap.

The common orthogonal map in \cref{eq:frame} cancels from each matched-date inner product, $\ip{G_tS_{i,t}}{G_tS_{j,t}}=\ip{S_{i,t}}{S_{j,t}}$. Overlapping targets can affect cities and the temporal law of the transformed histories, but do not directly change $D_{ij}$ for fixed original signals. A common random target alignment also does not automatically preserve independence between entire transformed histories.

\subsection{Random temporal matching with the cities held fixed}

Fix vector collections $(X_t)_{t=1}^T$ and $(Y_t)_{t=1}^T$, with $T\ge2$, and let $\pi$ be a uniform random permutation of $\{1,\ldots,T\}$. Define
\begin{equation}\label{eq:permutation}
 D_\pi=\frac1T\sum_t\ip{X_t}{Y_{\pi(t)}}.
\end{equation}
The cities and retained lengths are unchanged by the permutation. Put
\begin{equation}\label{eq:covariances}
 \begin{aligned}
 \Sigma_X&=\frac1T\sum_t(X_t-m_X)(X_t-m_X)^\top,\\
 \Sigma_Y&=\frac1T\sum_t(Y_t-m_Y)(Y_t-m_Y)^\top.
 \end{aligned}
\end{equation}

The idea is to keep both sets of vectors exactly as observed and scramble only which date of one is matched with which date of the other, so that any spread in $D_\pi$ comes from temporal matching alone.

\begin{theorem}[Exact conditional permutation moments]\label{thm:permutation}
With nonzero means and $C=\ip{c_X}{c_Y}$,
\begin{equation}\label{eq:permutationmoments}
 \E_\pi D_\pi=r_Xr_YC,\qquad
 \Var_\pi(D_\pi)=\frac{\tr(\Sigma_X\Sigma_Y)}{T-1}=:V.
\end{equation}
If $\tau>r_Xr_YC$, then
\begin{equation}\label{eq:cantelli}
 \Prob_\pi(D_\pi\ge\tau)
 \le\frac{V}{V+(\tau-r_Xr_YC)^2}.
\end{equation}
When $V>0$, for every $0<\delta<1$, with probability at least $1-\delta$,
\begin{equation}\label{eq:permcertificate}
 D_\pi\le r_Xr_YC+\sqrt{V\frac{1-\delta}{\delta}}.
\end{equation}
If $V=0$, $D_\pi=r_Xr_YC$ almost surely. All formulas for the moments remain valid with $\ip{m_X}{m_Y}$ if a mean vanishes.
\end{theorem}
\begin{proof}
Set $a_t=X_t-m_X$ and $b_t=Y_t-m_Y$, so both sequences sum to zero. The random part of $D_\pi$ is $T^{-1}\sum_ta_t^\top b_{\pi(t)}$, whose mean is zero. Uniform permutation gives
\[
 \E_\pi[b_{\pi(t)}b_{\pi(t)}^\top]=\Sigma_Y,
 \qquad
 \E_\pi[b_{\pi(t)}b_{\pi(s)}^\top]=-\frac{\Sigma_Y}{T-1}\quad(t\ne s).
\]
Expanding the square and using $\sum_ta_t=0$ yields
\begin{align*}
 \Var_\pi(D_\pi)
 &=\frac{1}{T(T-1)}\sum_ta_t^\top\Sigma_Ya_t\\
 &=\frac{\tr(\Sigma_X\Sigma_Y)}{T-1}.
\end{align*}
The one-sided second-moment inequality of Cantelli~\cite[Exercise~1.6.4]{durrett} gives \cref{eq:cantelli}; solving its upper bound for $\delta$ gives \cref{eq:permcertificate}.
\end{proof}

The mean under random matching is exactly the part of the correlation that the compressed summary sees, $r_Xr_YC$. An observed correlation far above that mean, measured in units of $\sqrt V$, would be unusual if the dates were matched at random; this is the sense in which the theorem benchmarks synchronized movement that cities cannot detect.

For an absolute-correlation rule and $\tau>|r_Xr_YC|$, Chebyshev's inequality also gives
\begin{equation}\label{eq:permabsolute}
 \Prob_\pi(|D_\pi|\ge\tau)
 \le\min\left\{1,\frac{V}{(\tau-|r_Xr_YC|)^2}\right\}.
\end{equation}
In these statements the only randomness is how the dates of the two collections are matched; the collections themselves are held fixed. Conditioning on the complete collections in this way is stronger than a generic high-probability bound evaluated after a city screen.

\begin{corollary}[Isotropic residual covariance]\label{cor:isotropic}
If $\Sigma_Y=(1-r_Y^2)\id_q/q$, then
\begin{equation}\label{eq:isotropicV}
 V=\frac{(1-r_X^2)(1-r_Y^2)}{q(T-1)}.
\end{equation}
For $r_Y<1$, this hypothesis requires $T\ge q+1$. If $r_Y=1$, the covariance and $V$ vanish and no such rank restriction is needed.
\end{corollary}
\begin{proof}
Use $\tr\Sigma_X=1-r_X^2$ in \cref{eq:permutationmoments}. An empirical covariance of $T$ centered observations has rank at most $T-1$, whereas a positive multiple of $\id_q$ has rank $q$.
\end{proof}

The covariance overlap $\tr(\Sigma_X\Sigma_Y)$, rather than asset count alone, controls accidental residual alignment. The isotropic formula displays a dimension gain only under its stated covariance assumption.

There is an exact consistency check with \cref{thm:northernnull}. Under the uniform northern-history law, sign reversal leaves each empirical covariance unchanged, so
\[
 \E_T^+\Sigma_X=\E_T^+\Sigma_Y
 =\frac{T-1}{qT}\id_q.
\]
The two covariances are independent. Averaging the general overlap formula in \cref{eq:permutationmoments} therefore gives
\begin{equation}\label{eq:permutationnullconsistency}
 \E_T^+V
 =\frac{\tr\bigl[(\E_T^+\Sigma_X)(\E_T^+\Sigma_Y)\bigr]}{T-1}
 =\frac{T-1}{qT^2},
\end{equation}
exactly the variance in \cref{eq:northernmoments}. This does not require each realized empirical covariance to be isotropic.

Random matching is a useful benchmark for histories with no synchronized temporal structure, but it is not a default model of market signals. An unrestricted permutation destroys serial dependence and can change how a signal relates to the original dated sequence of targets. To use this randomization as an empirical test, one must justify exchangeability for the hypothesis being tested. The statistic itself belongs to the classical permutation-sum framework~\cite{hoeffding}; we have computed its finite-sample moments directly.

\subsection{Independent dates with controlled residual dependence}

Now suppose the paired vectors $(X_t,Y_t)$ are independent and identically distributed over dates. Dependence between $X_t$ and $Y_t$ at a date is allowed. Define
\begin{equation}\label{eq:kappa}
 \mu_X=\E X_t,\quad\mu_Y=\E Y_t,\quad
 \kappa=\E\ip{X_t-\mu_X}{Y_t-\mu_Y}.
\end{equation}
Sample means and cities are computed from the same evaluation dates. The discarded sample residual cross-moment from \cref{eq:residuals,eq:decomposition} is
\begin{equation}\label{eq:sampleK}
 K_T=\frac1T\sum_{t=1}^T\ip{X_t-m_X}{Y_t-m_Y}
 =D_T-\ip{m_X}{m_Y}.
\end{equation}
Thus $\kappa$ is the population counterpart of exactly the quantity controlled deterministically in \cref{thm:sharp}. The next result concentrates $K_T$ around its expectation $(1-1/T)\kappa$; the factor $1-1/T$ comes from centering by sample means while using divisor $T$.

The earlier null models take the two signals to be unrelated. Here they may share movement on every date, and the question is whether a long history pins down how much.

\begin{theorem}[Concentration of the discarded residual cross-moment]\label{thm:concentration}
For iid paired unit vectors and $0<\delta<1$, with probability at least $1-\delta$,
\begin{equation}\label{eq:mcdiarmid}
 \left|K_T-\left(1-\frac1T\right)\kappa\right|
 \le\sqrt{\frac{18\log(2/\delta)}{T}}.
\end{equation}
If both sample means are nonzero, $D_T=r_Xr_YC_T+K_T$.
\end{theorem}
\begin{proof}
By \cref{eq:sampleK} and independence across dates,
\[
 \E\ip{m_X}{m_Y}
 =\frac1T\E\ip{X_t}{Y_t}+\left(1-\frac1T\right)\ip{\mu_X}{\mu_Y},
\]
so $\E K_T=(1-1/T)\kappa$. Replacing one paired observation changes $D_T$ by at most $2/T$. Each sample mean changes by at most $2/T$, and all sample mean norms are at most one; hence their inner product changes by at most $4/T$. Thus $K_T$ has bounded-difference constants $6/T$. McDiarmid's inequality~\cite{mcdiarmid} gives
\begin{align*}
 \Prob(|K_T-\E K_T|>z)
 &\le2\exp\left(-\frac{2z^2}{T\cdot(6/T)^2}\right)\\
 &=2\exp\left(-\frac{Tz^2}{18}\right).
\end{align*}
Here $18=6^2/2$ is independent of both $q$ and $d$. Set the right-hand side equal to $\delta$.
\end{proof}

The bound is dimension-free and deliberately conservative. To reduce its stochastic error term to a tolerance $\varepsilon>0$, it suffices to have
\begin{equation}\label{eq:concentrationsamplesize}
 T\ge\left\lceil\frac{18\log(2/\delta)}{\varepsilon^2}\right\rceil.
\end{equation}
For $\delta=0.05$ and $\varepsilon=0.1$, this is $T\ge6{,}640$ independent dates: about 26.3 years at 252 observations per year, or 18.2 years at 365. The guarantee can therefore be too conservative for a short evaluation history. The gain from dimension in \cref{cor:isotropic} is specific to its conditional permutation model. Marginal isotropy alone does not carry that gain over to paired observations whose same-date dependence is unrestricted.

A possible Bernstein refinement~\cite{blm} would exploit a small paired variance $\Var(\ip{X_t-\mu_X}{Y_t-\mu_Y})$, while separately controlling the sample-mean correction $\ip{m_X-\mu_X}{m_Y-\mu_Y}$. Both ingredients would be needed to improve the resulting bound on $K_T$.

\begin{corollary}[A residual-dependent novelty certificate]\label{cor:probnovelty}
If $|\kappa|\le\kappa_0$, then on the event in \cref{thm:concentration},
\begin{equation}\label{eq:probabsolute}
 |D_T|\le r_Xr_Y|C_T|
 +\left(1-\frac1T\right)\kappa_0
 +\sqrt{\frac{18\log(2/\delta)}{T}}.
\end{equation}
A one-sided upper bound follows by removing absolute values on $D_T$ and $C_T$ and assuming only $\kappa\le\kappa_0$.
\end{corollary}

The three terms have distinct sources: what the cities and retained lengths show, how much synchronized movement the two signals are assumed to share, and sampling error.

Independence between the two signal streams is sufficient, but not necessary, for $\kappa=0$. The bound $\kappa_0$ is an assumption about exactly the synchronized temporal variation that cities discard, and shared factors may make $\kappa$ large. An application must justify or estimate $\kappa_0$; it cannot be inferred from city separation.

\subsection{Nonzero population means and the limiting separation margin}

For signals with persistent nonzero means, the limiting ordering is determined by a population margin. This result permits both dependence across dates and shared movement between signals.

The uniform null has no skill and no shared movement. A real library may have both: means that do not vanish and signals that respond to common drivers. The next theorem identifies what then decides, over a long history, whether the city angle is conservative.

\begin{theorem}[Stationary ergodic population margin]\label{thm:populationordering}
Let $(X_t,Y_t)$ be a jointly stationary ergodic process of paired unit vectors in a fixed dimension, with means $\mu_X,\mu_Y$ of lengths $\rho_X,\rho_Y>0$. Define
\begin{equation}\label{eq:populationmargin}
 \begin{aligned}
 C_*&=\frac{\mu_X^\top\mu_Y}{\rho_X\rho_Y},\qquad
 \kappa=\E\ip{X_t-\mu_X}{Y_t-\mu_Y},\\
 g&=(1-\rho_X\rho_Y)C_*-\kappa.
 \end{aligned}
\end{equation}
Then $D_T-C_T\to-g$ almost surely. If $g>0$, eventually $\theta_T>\phi_T$ almost surely; if $g<0$, eventually $\theta_T<\phi_T$ almost surely. At $g=0$ ergodicity alone gives no ordering conclusion. If both population mean ICs are strictly positive, northern selection has probability tending to one and does not change these limiting conclusions.
\end{theorem}
\begin{proof}
Birkhoff's ergodic theorem, applied to each coordinate and to $\ip{X_t}{Y_t}$, gives $m_X\to\mu_X$, $m_Y\to\mu_Y$, and $D_T\to\mu_X^\top\mu_Y+\kappa$ almost surely~\cite[Section~6.2]{durrett}. The sample means are therefore nonzero eventually, and continuity of normalization gives $C_T\to C_*$. Subtraction yields $D_T-C_T\to-g$. Strict margins fix the eventual sign, and $\arccos$ is decreasing. Strictly positive population projections on $e$ imply eventual northern selection almost surely; conditioning on an event whose probability tends to one preserves the corresponding limiting probabilities.
\end{proof}

The margin $g$ sets two forces against each other. Shared movement between the signals, measured by $\kappa$, pulls their histories together. Against it stands the allowance $(1-\rho_X\rho_Y)C_*$, which short mean vectors enlarge and which is positive only when the population cities are less than $90^\circ$ apart. Whichever is larger settles the ordering for all sufficiently long histories, and \cref{ex:sharedpersistent} shows shared movement prevailing even when both signals have positive IC on every date.

Thus the relevant condition is $\kappa<(1-\rho_X\rho_Y)C_*$. Independence between the two streams is sufficient for $\kappa=0$; if also $C_*>0$ and $\rho_X\rho_Y<1$, preservation has probability tending to one. More generally, positive residual dependence is compatible with preservation when it stays below the margin. Ergodicity alone supplies neither a convergence rate nor a central limit theorem.

\Needspace{8\baselineskip}
\begin{corollary}[Residual fluctuations under independent paired dates]\label{cor:populationclt}
In the setting of \cref{thm:populationordering}, assume additionally that the paired observations are iid over dates. If $\sigma_K^2:=\Var(\ip{X_t-\mu_X}{Y_t-\mu_Y})>0$, then
\begin{equation}\label{eq:populationresidualclt}
 \sqrt T(K_T-\kappa)\dto N(0,\sigma_K^2).
\end{equation}
For independent streams, $\kappa=0$ and $\sigma_K^2=\tr(\Sigma_X^{\rm pop}\Sigma_Y^{\rm pop})$. When this trace is positive, $\Prob(K_T\le0)\to1/2$.
\end{corollary}
\begin{proof}
Put $h_t=\ip{X_t-\mu_X}{Y_t-\mu_Y}$. Direct expansion gives
\[
 K_T-\kappa=\frac1T\sum_{t=1}^T(h_t-\kappa)
 -\ip{m_X-\mu_X}{m_Y-\mu_Y}.
\]
Each centered mean is $O_{\Prob}(T^{-1/2})$, so the final term is $O_{\Prob}(T^{-1})$. The scalar iid central limit theorem proves the claim. Independence between streams gives the covariance-trace formula, and its strict positivity gives a continuous, nondegenerate Gaussian limit and hence sign balance.
\end{proof}

\begin{example}[Persistent shared movement can reverse the limiting ordering]\label{ex:sharedpersistent}
In dimension $q\ge4$, choose orthonormal $e,v,u,w$ and set
\[
 c_X=\tfrac{\sqrt3}{2}e+\tfrac12v,\quad
 c_Y=\tfrac{\sqrt3}{2}e-\tfrac12v,\quad
 u_X=u,\quad u_Y=0.8u+0.6w.
\]
Let $Z_t\in\{-1,1\}$ be a stationary symmetric Markov chain with probability $(1+a)/2$ of retaining its sign, $0\le a<1$. Its lag correlation is $a^h$ and it is ergodic. With $r=0.1$, put
\[
 X_t=rc_X+\sqrt{1-r^2}Z_tu_X,\qquad
 Y_t=rc_Y+\sqrt{1-r^2}Z_tu_Y.
\]
Both vectors are unit and have IC $r\sqrt3/2>0$ at every date. Their means converge to $rc_X,rc_Y$, so $C_T\to0.5$. Yet $D_T=0.01(0.5)+0.99(0.8)=0.797$ for every $T$. Thus $\kappa=0.792$, $g=-0.297$, and separation eventually fails. Taking $a$ near one makes serial dependence persistent. Taking $a=0$ leaves the failure intact with independent dates: the shared movement changes the limiting margin.
\end{example}

\Cref{thm:populationordering} assumes fixed nonzero population means, and this is a substantive assumption. A small empirical retained length is compatible both with a zero population mean and with a stable nonzero one; it proves neither. When means are near zero, normalization is sensitive to sampling variation, and when means drift, the cities may not converge at all. The transition model in \cref{prop:localalternatives} illustrates the case in which mean and noise have comparable scales, but only under its own iid triangular-array assumptions; it is not a limit theorem for dependent data. The exact finite-sample residual identity requires neither stationarity nor convergence.

\subsection{Simultaneous comparisons and the meaning of confidence}

Let $J\ge2$ signal processes be fixed before an independent evaluation sample, and assume the joint collection of transformed signal vectors is iid over evaluation dates. Dependence among signals at each date is unrestricted. There are $M=J(J-1)/2$ unordered pairs. A union bound in \cref{thm:concentration} replaces the error term by
\begin{equation}\label{eq:uniformerror}
 \varepsilon_{T,J,\delta}=\sqrt{\frac{18\log(2M/\delta)}T}.
\end{equation}
With probability at least $1-\delta$, the corresponding bound holds for every pair simultaneously. Given valid pairwise bounds on $\kappa$, this bounds the probability that some pair passes the sufficient screen and yet violates the claimed empirical correlation threshold. It does not license constructing new signals indefinitely against the same evaluation sample~\cite{white}.

This guarantee is also distinct from a conditional error rate among the pairs that pass a rare screen. In general,
\[
 \Prob(\text{false certificate and screen passes})\le\delta
\]
does not imply $\Prob(\text{false certificate}\mid\text{screen passes})\le\delta$; this is the selection effect analyzed for selected parameters in~\cite{benjaminiyekutieli}. \Cref{thm:permutation} has its own explicit conditional probability model.

Financial observations are typically dependent across dates, especially when signals are persistent and target horizons overlap. Two of our results allow such dependence: \cref{thm:dependentordering}, an ordering bound for independent rotationally symmetric histories, and \cref{thm:populationordering}, the population-margin limit. Neither extends the McDiarmid certificate in \cref{thm:concentration} to dependent dates. Extending that certificate would require an explicit dependence assumption and a matching concentration argument; substituting an unspecified effective sample size is not such an argument.

\section{Progressive compression by time blocks}\label{sec:blocks}

Instead of one temporal mean per signal, one can store a mean for each of several time blocks shared by all signals. This keeps part of each signal's evolution while remaining a compressed representation.

Let $\mathcal P$ be a partition of the dates into nonempty blocks. For a block $b$, set $w_b=|b|/T$ and
\begin{equation}
 m_{i,b}=\frac1{|b|}\sum_{t\in b}X_{i,t}=r_{i,b}c_{i,b}
 \quad\text{when }r_{i,b}>0.
\end{equation}
A zero block mean is stored as a zero vector without assigning it a direction. Define
\begin{align}
 A_{ij}^{\mathcal P}
 &=\sum_{b\in\mathcal P}w_b\ip{m_{i,b}}{m_{j,b}},\label{eq:blockcenter}\\
 E_i^{\mathcal P}
 &=1-\sum_{b\in\mathcal P}w_b\norm{m_{i,b}}^2.\label{eq:blockenergy}
\end{align}

When the single-mean certificate is inconclusive, the natural next step is to publish a little more: one mean vector per block of dates. The next theorem shows that the resulting interval for $D_{ij}$ always contains the true value and can only shrink as blocks are subdivided.

\begin{theorem}[Nested correlation intervals]\label{thm:blocks}
For every partition $\mathcal P$,
\begin{equation}\label{eq:blockinterval}
 D_{ij}\in I_{\mathcal P}:=
 \left[A_{ij}^{\mathcal P}-\sqrt{E_i^{\mathcal P}E_j^{\mathcal P}},
 A_{ij}^{\mathcal P}+\sqrt{E_i^{\mathcal P}E_j^{\mathcal P}}\right].
\end{equation}
If $\mathcal Q$ refines $\mathcal P$, then $I_{\mathcal Q}\subseteq I_{\mathcal P}$. For the singleton partition, $I_{\mathcal Q}=\{D_{ij}\}$.
\end{theorem}
\begin{proof}
Equip histories with inner product
\[
 \ip fg_T=\frac1T\sum_t\ip{f_t}{g_t}.
\]
Block averaging is the orthogonal projection $P_{\mathcal P}$ onto block-constant histories. The projected inner product is $A_{ij}^{\mathcal P}$, and the residual squared norms are $E_i^{\mathcal P}$. Orthogonality and Cauchy--Schwarz give \cref{eq:blockinterval}.

For a refinement put
\[
 V_i=(P_{\mathcal Q}-P_{\mathcal P})X_i,\qquad
 W_i=(\id-P_{\mathcal Q})X_i.
\]
The components are orthogonal. The interval center changes by $\ip{V_i}{V_j}_T$, while its old radius is
\[
 B_{\mathcal P}
 =\sqrt{(\norm{V_i}_T^2+\norm{W_i}_T^2)
 (\norm{V_j}_T^2+\norm{W_j}_T^2)}.
\]
Cauchy--Schwarz in $\R^2$ and then in the history space give
\begin{align*}
 B_{\mathcal P}
 &\ge\norm{V_i}_T\norm{V_j}_T+\norm{W_i}_T\norm{W_j}_T\\
 &\ge\left|\ip{V_i}{V_j}_T\right|+B_{\mathcal Q}.
\end{align*}
The new center shift plus new radius therefore fits inside the old interval. At the singleton partition the projection is the identity, so the residuals vanish.
\end{proof}

This result gives a deterministic resolution hierarchy. Start with one mean vector per signal; refine the shared partition when a pair remains inconclusive; compare full histories only when necessary. The interval upper endpoint certifies a signed upper bound, while an interval contained in $[-\tau,\tau]$ certifies absolute novelty. Refinement can only strengthen these certificates.

The natural compressed object is therefore the mean vector $m_i=r_ic_i$, a point in the closed unit ball, rather than the city alone. The city is its direction, and its length determines the residual energy. The proof works for any nested orthogonal projections; we use time blocks because each block mean can be read as the city of a shorter period. The bounds require that the compared signals share one partition and one frame convention. \Cref{sec:blockexperiment} measures how fine the partition must be for the competition library.

\section{Stability and wandering of cities}\label{sec:stability}

Retained length also governs directional stability. For $m\ne0$, normalization $n(m)=m/\norm m$ has derivative
\begin{equation}\label{eq:derivative}
 Dn_m[h]=\frac{(\id-cc^\top)h}{\norm m},\qquad c=m/\norm m.
\end{equation}
A small retained length therefore amplifies transverse perturbations of the mean. The same quantity thus governs both problems: how much a city reveals about novelty, and how far a city wanders.

\subsection{Persistent positive IC and expanding-window motion}

Consider one infinite sequence $(X_t)_{t\ge1}$ of unit vectors in the specified common frame, and suppress the signal index. Set
\begin{equation}\label{eq:expandingnotation}
 m_T=\frac1T\sum_{t=1}^T X_t,\qquad r_T=\norm{m_T},\qquad
 \alpha_T=\frac1T\sum_{t=1}^T\IC_t=\ip{m_T}e.
\end{equation}
Whenever $r_T>0$, put $c_T=m_T/r_T$. Here a history of length $T$ means the first $T$ observations. The coordinates of earlier observations are kept fixed when the history is extended. A rolling-window statement requires the corresponding assumption on each rolling window.

We measure predictivity by target alignment. Positive empirical average IC, a positive population mean IC, and a positive empirical margin at every length are distinct conditions. The exact identity
\begin{equation}\label{eq:citylatitude}
 \alpha_T=r_T\ip{c_T}e
\end{equation}
shows that $\alpha_T>0$ is equivalent to a well-defined city in the open northern hemisphere. For the unit stacked history $x_T=T^{-1/2}(X_1,\ldots,X_T)$ and North-Pole history $e_T=T^{-1/2}(e,\ldots,e)$, the whole-history angle is
\begin{equation}\label{eq:targethistoryangle}
 \Theta_T=\arccos\ip{x_T}{e_T}=\arccos\alpha_T.
\end{equation}
Thus positive average IC is also equivalent to $\Theta_T<\pi/2$. This angle is not the average of the individual polar angles.

\begin{definition}[Persistent positive IC]\label{def:persistentIC}
A sequence has persistent positive IC if there are $a\in(0,1]$ and $T_0\ge1$ such that $\alpha_T\ge a$ for all $T\ge T_0$. Equivalently, $\liminf_{T\to\infty}\alpha_T>0$.
\end{definition}

\Needspace{15\baselineskip}
A city computed from an expanding history moves as observations arrive. How far can it move, and can it leave the target's hemisphere, if the signal keeps a positive average IC?

\begin{theorem}[Persistent IC confines cities and controls their motion]\label{thm:predictivemotion}
Suppose $\alpha_T\ge a>0$ for every $T\ge T_0$. Then the cities exist and, for every such $T$,
\begin{equation}\label{eq:predictivecap}
 \boxed{r_T\ge a,\qquad
 d_{\sphere}(c_T,e)\le\arccos a<\pi/2,}
\end{equation}
where $d_{\sphere}(x,y)=\arccos\ip xy$. Successive cities satisfy the chordal motion bound
\begin{equation}\label{eq:predictivestep}
 \boxed{\norm{c_{T+1}-c_T}\le\frac{1}{a\sqrt{T(T+1)}}.}
\end{equation}
No independence, stationarity, or distributional assumption is required.
\end{theorem}
\begin{proof}
Since $\norm{X_t}=1$, we have $r_T\le1$. From $\alpha_T=\ip{m_T}e$ and Cauchy--Schwarz,
\[
 r_T\ge\alpha_T\ge a,\qquad
 \ip{c_T}e=\frac{\alpha_T}{r_T}\ge\alpha_T\ge a.
\]
These inequalities give \cref{eq:predictivecap}. To control motion, write $S_T=\sum_{t=1}^T X_t=Tr_Tc_T$. The identity
\begin{equation}\label{eq:sumstepidentity}
 \norm{S_{T+1}-S_T}^2
 =\bigl(\norm{S_{T+1}}-\norm{S_T}\bigr)^2
 +\norm{S_{T+1}}\norm{S_T}\norm{c_{T+1}-c_T}^2
\end{equation}
has left-hand side $\norm{X_{T+1}}^2=1$. Dropping its first nonnegative term and using $\norm{S_T}\ge aT$ and $\norm{S_{T+1}}\ge a(T+1)$ proves \cref{eq:predictivestep}.
\end{proof}

For a platform the theorem means that the city of a persistently predictive signal changes slowly: after $T$ observations with margin $a$, one further observation moves it by less than $1/(aT)$.

The theorem says that the mean cannot collapse and that the city stays inside a fixed spherical cap. The step bound is of order $T^{-1}$, which is not summable over $T$, so it does not imply convergence. The retained-length and cap bounds are individually sharp: alternating $ae+\sqrt{1-a^2}\,u$ and $ae-\sqrt{1-a^2}\,u$, with $u\perp e$ unit, gives $r_T=a$ at every even $T$; a constant sequence $X_t=ae+\sqrt{1-a^2}\,u$ attains the cap boundary.

The same margin also supplies a sufficient assumption for the novelty problem. At a common history length, if two signals both have average IC at least $a>0$, then $r_i,r_j\ge a$ and \cref{cor:distance} gives
\begin{equation}\label{eq:predictivenovelty}
 \sin\frac{\theta_{ij}}2\ge a\sin\frac{\phi_{ij}}2.
\end{equation}
This is a deterministic separation guarantee for the predictive class. Its strength depends on the size of the IC margin; a small positive margin gives only a weak lower bound on history separation.

\subsection{Two angular conditions and their relation to IC}

Let $\theta_t=d_{\sphere}(X_t,e)\in[0,\pi]$, and let
\[
 \mathcal E=\{x\in\sphere(H):\ip xe=0\}
\]
be the equator. Angular distance to this set is
\begin{equation}\label{eq:equatordistance}
 d_{\sphere}(X_t,\mathcal E)
 :=\inf_{z\in\mathcal E}d_{\sphere}(X_t,z)
 =\left|\frac\pi2-\theta_t\right|.
\end{equation}
The phrase ``closer to the North Pole than to the equator'' can be read in two ways, and the two readings are not equivalent. The first compares the average polar angle with the pole-to-equator distance $\pi/2$. The second compares the average polar angle with the average angular distance of the observations from the equator.

\begin{example}[An average polar angle below $\pi/2$ does not imply positive IC]\label{ex:angularIC}
Fix a unit $u\perp e$ and put $v=-e/2+(\sqrt3/2)u$, so that $d_{\sphere}(v,e)=2\pi/3$. Repeat the ten-observation cycle consisting of three copies of $e$ followed by seven copies of $v$. Every prefix satisfies
\begin{equation}\label{eq:angularcounterexample}
 \frac1T\sum_{t=1}^T\theta_t\le\frac{7\pi}{15}<\frac\pi2,
 \qquad \alpha_{10k}=\frac{3-7/2}{10}=-\frac1{20}
 \quad(k\ge1).
\end{equation}
To verify the prefix bound, write $T=10k+j$, $0\le j\le9$. The number of copies of $v$ is $7k+\max\{j-3,0\}\le7T/10$. The IC sum is
\begin{equation}\label{eq:angularprefixsum}
 T\alpha_T=-\frac{k}{2}+\min\{j,3\}-\frac12\max\{j-3,0\}.
\end{equation}
At $T=63$ this is zero, so the city lies on the equator. It is strictly negative for every $T\ge64$: if $k=6$ and $j\ge4$, it is at most $-1/2$; if $k\ge7$, it is at most $-k/2+3<0$. Thus the average polar angle never exceeds $84^\circ$, yet the city lies in the southern hemisphere for every sufficiently long prefix. The failure persists under a uniform angular margin below $90^\circ$.
\end{example}

For the second interpretation, define the angular advantage
\begin{equation}\label{eq:angulargap}
 g_T=\frac1T\sum_{t=1}^T
 \bigl[d_{\sphere}(X_t,\mathcal E)-d_{\sphere}(X_t,e)\bigr].
\end{equation}
All angular quantities in the following bound are in radians.

\begin{proposition}[Angular advantage implies positive average IC]\label{prop:angulargap}
For every finite unit-vector history,
\begin{equation}\label{eq:gapIC}
 \alpha_T\ge\frac{2g_T}{\pi}.
\end{equation}
In particular, $g_T>0$ guarantees a northern city. If $g_T\ge\eta>0$ for all $T\ge T_0$, then \cref{thm:predictivemotion} applies with $a=2\eta/\pi$.
\end{proposition}
\begin{proof}
For $0\le\theta\le\pi/2$, concavity of cosine on this interval gives
\[
 \cos\theta\ge1-\frac{2\theta}{\pi}
 \ge1-\frac{4\theta}{\pi}
 =\frac2\pi\left(\left|\frac\pi2-\theta\right|-\theta\right).
\]
For $\pi/2\le\theta\le\pi$, the right-hand side equals $-1$ and is again no larger than $\cos\theta$. Averaging this pointwise inequality gives \cref{eq:gapIC}. Equality holds for histories supported on the two poles, so the coefficient $2/\pi$ cannot be increased in a universal linear lower bound of this form for $g_T>0$.
\end{proof}

This angular advantage is sufficient but not necessary for positive IC. A constant sequence at polar angle $60^\circ$ has IC $1/2$, although its distance to the equator is only $30^\circ$ and its angular advantage is negative. The class with persistent positive IC is therefore broader than the class with a persistent positive angular advantage.

Strict positivity at every length must also be distinguished from a fixed positive margin. For example, set $X_1=e$ and, for $k\ge1$, $X_{2k}=e$ and $X_{2k+1}=-e$. Then $\alpha_T=r_T$ equals $1/T$ for odd $T$ and $2/T$ for even $T$, and $g_T=(\pi/2)\alpha_T>0$ at every length. Nevertheless $r_T\to0$. A city can remain at the North Pole even as its retained length vanishes.

\subsection{From confinement to convergence}

\begin{example}[Uniform predictivity and angular advantage without convergence]\label{ex:nonconvergentpredictive}
Fix a unit $u\perp e$ and an angle $0<\beta<\pi/4$, and define
\[
 x_+=\cos\beta\,e+\sin\beta\,u,\qquad
 x_-=\cos\beta\,e-\sin\beta\,u.
\]
Let block $k$ have length $L_k=2^{2^k}$, using $x_+$ on odd blocks and $x_-$ on even blocks. Every observation has IC $\cos\beta>0$ and angular advantage $\pi/2-2\beta>0$, so both margins hold on every prefix and indeed on every finite window.

Write $N_k=\sum_{j=1}^kL_j$ and let $x^{(k)}$ denote the vector used in block $k$. Since $N_{k-1}/L_k\to0$,
\[
 \norm{m_{N_k}-x^{(k)}}
 \le\frac{2N_{k-1}}{N_k}\longrightarrow0.
\]
Normalization is continuous at each unit vector $x_\pm$. Consequently the cities at odd block endpoints approach $x_+$, while those at even block endpoints approach $x_-$. The city does not converge, despite uniform positive IC and a uniform positive angular advantage. Each successive step still obeys \cref{eq:predictivestep}.
\end{example}

Convergence to the North Pole requires control of the transverse average. Decompose
\begin{equation}\label{eq:transversemean}
 m_T=\alpha_Te+u_T,\qquad
 u_T=(\id-ee^\top)m_T\perp e.
\end{equation}

\begin{corollary}[Transverse cancellation gives North-Pole convergence]\label{cor:northconvergence}
Suppose $\alpha_T\ge a>0$ for all sufficiently large $T$ and $\norm{u_T}\to0$. Then $c_T\to e$, with the quantitative bound
\begin{equation}\label{eq:northconvergence}
 \tan d_{\sphere}(c_T,e)
 =\frac{\norm{u_T}}{\alpha_T}
 \le\frac{\norm{u_T}}{a}.
\end{equation}
If also $\alpha_T\to\alpha_*>0$, then $r_T\to\alpha_*$.
\end{corollary}
\begin{proof}
Orthogonality in \cref{eq:transversemean} gives $r_T^2=\alpha_T^2+\norm{u_T}^2$. The polar angle is strictly less than $\pi/2$ and has cosine $\alpha_T/r_T$ and sine $\norm{u_T}/r_T$, proving \cref{eq:northconvergence}. Both convergence claims follow.
\end{proof}

No limit of $\alpha_T$ is needed for the directional conclusion. More generally, if $\alpha_T\to\alpha_*>0$ and $u_T\to u_*$, then
\begin{equation}\label{eq:generalnorthlimit}
 c_T\longrightarrow
 \frac{\alpha_*e+u_*}{\sqrt{\alpha_*^2+\norm{u_*}^2}}.
\end{equation}
A predictive signal can therefore have a stable city away from the North Pole. Whether the transverse average cancels, and hence whether the city converges to the North Pole, depends on the chosen frame convention. The IC margin does not: it is invariant under the residual transformations of \cref{prop:gauge}.

\Needspace{19\baselineskip}
\subsection{Nonzero population means and isotropic noise}

The deterministic results say where a city can be. Under a simple sampling model one can also say whether it settles, and how quickly.

\begin{proposition}[Nonzero means and the isotropic zero-mean regime]\label{prop:stability}
Let $X_t$ be iid unit vectors in $\R^q$, with $\mu=\E X_t$, $\rho=\norm\mu$, and $\Sigma=\Cov(X_t)$. Then
\begin{equation}\label{eq:meanlength}
 \E\norm{m_T}^2=\rho^2+\frac{1-\rho^2}{T}.
\end{equation}
If $\rho>0$, then $c_T\to c_*=\mu/\rho$ almost surely, and
\begin{equation}\label{eq:directionclt}
 \sqrt T(c_T-c_*)\dto
 N\left(0,\frac{(\id-c_*c_*^\top)\Sigma(\id-c_*c_*^\top)}{\rho^2}\right).
\end{equation}
If $X_t$ is uniform on $\sphere^{q-1}$ with $q\ge2$, then $m_T\to0$ almost surely, $\E r_T^2=1/T$, the city is uniform for every fixed $T$, and
\begin{equation}\label{eq:chisquare}
 T r_T^2\dto\chi_q^2/q.
\end{equation}
In this isotropic zero-mean model, $(c_T)$ does not converge in probability, even to a random direction.
\end{proposition}
\begin{proof}
Expanding the squared norm of the sample mean gives $1/T$ from the diagonal terms and $(1-1/T)\norm\mu^2$ from the distinct-date terms, proving \cref{eq:meanlength}. If $\rho>0$, the strong law and continuity of normalization give almost-sure convergence; the multivariate central limit theorem and \cref{eq:derivative} give \cref{eq:directionclt}. Cities are defined for all sufficiently large $T$ almost surely in this case; arbitrary definitions at earlier zero means do not change the limit.

Under uniform spherical sampling, the sample sum is rotationally invariant and nonzero almost surely for each fixed $T$. Its direction is therefore uniform. The strong law gives $m_T\to0$, and the central limit theorem gives $\sqrt Tm_T\dto N(0,\id_q/q)$, implying \cref{eq:chisquare}.

For the final claim, split the first $2T$ observations into two independent blocks. Their normalized sums converge jointly to independent isotropic nondegenerate Gaussians $Z_1,Z_2$. Consequently $(c_T,c_{2T})$ converges in distribution to
\[
 \left(\frac{Z_1}{\norm{Z_1}},\frac{Z_1+Z_2}{\norm{Z_1+Z_2}}\right).
\]
The distance between these limiting directions has positive probability of exceeding some fixed positive number. Hence $(c_T)$ is not Cauchy in probability, which rules out convergence in probability to any random vector.
\end{proof}

The contrast shows why retained length is the quantity to watch. With a nonzero population mean of length $\rho$, the city's sampling error is of order $1/(\rho\sqrt T)$, so a short mean vector means a noisy city. With isotropic zero-mean noise, the city is a uniformly random direction at every history length, however long.

In particular, if the iid population has positive mean IC $p=\ip\mu e>0$, the strong law gives $\alpha_T\to p$ almost surely. For every fixed $a\in(0,p)$, \cref{thm:predictivemotion} therefore applies after an almost surely finite, possibly sample-dependent $T_0$. If also $(\id-ee^\top)\mu=0$, then the city converges to $e$ by \cref{cor:northconvergence}. Population predictivity thus yields an eventual empirical margin under this sampling model; it does not require a positive IC on every date or every short prefix.

The nonzero-mean conclusion is the standard asymptotic theory of the sample mean direction~\cite{hendriks}. We include that conclusion because its denominator, the population mean length $\rho$, is the same quantity that controls the compression problem. The chi-square limit in \cref{eq:chisquare} is the classical Rayleigh statistic for testing uniformity~\cite{mardiajupp}. The city itself always has unit norm; what vanishes under zero-mean noise is the unnormalized average, and with it the retained length.

\subsection{Zero IC, zero mean, and isotropy are different assumptions}

A signal with zero average target projection can have a nonzero transverse mean. Its city may be stable while carrying no average IC. Conversely, a zero vector mean does not imply isotropy. If $\sqrt Tm_T\dto Z$ for a nondegenerate anisotropic Gaussian $Z$, then
\begin{equation}
 c_T\dto Z/\norm Z,
\end{equation}
which is generally not uniform on the sphere. To claim that noise visits all directions, one needs an isotropic model, or at least one with sufficiently broad support; zero IC alone is not enough.

\subsection{Expanding and rolling windows}

For a rolling window of fixed width $W$, put
\begin{equation}
 m_{t,W}=\frac1W\sum_{s=t-W+1}^tX_s.
\end{equation}
If each of two successive rolling windows has average IC at least $a>0$, its city obeys the cap and retained-length bounds in \cref{eq:predictivecap}. For $c_{t,W}=m_{t,W}/\norm{m_{t,W}}$, the corresponding motion bound is
\begin{equation}\label{eq:rollingpredictivestep}
 \norm{c_{t+1,W}-c_{t,W}}\le\frac{2}{aW}.
\end{equation}
Indeed, the unnormalized window sums differ by $X_{t+1}-X_{t-W+1}$, of norm at most two, and each has norm at least $aW$; apply the identity used in \cref{eq:sumstepidentity}. With fixed $W$, this bound does not decrease as calendar time advances. A positive margin on expanding prefixes does not imply the same margin on rolling windows.

Under iid zero-mean sampling, $\E\norm{m_{t,W}}^2=1/W$ at every eligible date. The distribution is stationary; it does not shrink merely because calendar time advances. Adjacent windows overlap and their cities are dependent.

\begin{proposition}[Full-sphere visits for isotropic rolling cities]\label{prop:rolling}
Under iid uniform sampling on $\sphere^{q-1}$, $q\ge2$, and fixed $W$, the rolling city visits every nonempty open subset of the sphere infinitely often almost surely.
\end{proposition}
\begin{proof}
At dates $W,2W,3W,\ldots$, the windows are disjoint, so their city directions are independent and uniform. Every nonempty open set has positive surface measure and is visited infinitely often by the second Borel--Cantelli lemma~\cite[Section~2.3]{durrett}. Apply this statement simultaneously to a countable basis for the sphere's topology. Every nonempty open set contains a basis element.
\end{proof}

Under an explicit model, this makes ``spherical wandering'' precise. It does not identify the motion of a rolling city with spherical Brownian motion, and it does not cover dependence, moving population means, or frame transport, each of which needs additional modeling. \Cref{sec:transition} gives an explicit sampling model for the transition regime, in which population means and sampling noise have comparable magnitude.

\subsection{A specified transition regime}\label{sec:transition}

Mean directions and sampling fluctuations contribute at the same scale when $\norm{\mu_T}$ is of order $T^{-1/2}$. The following triangular-array formulation gives explicit sufficient conditions. Here the population law may vary with the row length $T$; this is a sequence of sampling problems, rather than successive prefixes of one fixed process.

\Needspace{14\baselineskip}
Between a clear nonzero mean and pure noise lies the case most relevant to weak signals: a mean that is real but of the same order as its own sampling error.

\begin{proposition}[Local mean and noise on the same scale]\label{prop:localalternatives}
Fix $q\ge2$. For each $T$, let $X_{T,1},\ldots,X_{T,T}$ be iid unit vectors in $\R^q$, with mean $\mu_T$ and covariance $\Sigma_T$. Suppose
\[
 \sqrt{T}\mu_T\longrightarrow h,\qquad
 \Sigma_T\longrightarrow\Sigma,
\]
where $\Sigma$ necessarily satisfies $\tr\Sigma=1$ and need not be invertible. Put $m_T=T^{-1}\sum_{t=1}^T X_{T,t}$ and $r_T=\norm{m_T}$, and define $c_T=m_T/r_T$ with an arbitrary unit value when $m_T=0$. If $Z\sim N(0,\Sigma)$, then
\begin{equation}\label{eq:localalternatives}
 \bigl(\sqrt{T}r_T,c_T\bigr)
 \dto\left(\norm{h+Z},\frac{h+Z}{\norm{h+Z}}\right).
\end{equation}
\end{proposition}
\begin{proof}
For any fixed $v\in\R^q$, the centered row summands
\[
 \frac{\ip{v}{X_{T,t}-\mu_T}}{\sqrt{T}}
\]
are bounded in absolute value by $2\norm v/\sqrt{T}$. Their Lindeberg condition is therefore automatic, and the variance of their sum is $v^\top\Sigma_Tv\to v^\top\Sigma v$. If the limiting variance is positive, apply the scalar triangular-array central limit theorem~\cite[Section~3.4]{durrett}; if it is zero, Chebyshev's inequality gives convergence to zero in probability. The Cram\'er--Wold argument~\cite[Section~3.10]{durrett} and the assumed mean scaling therefore give
\[
 \sqrt{T}(m_T-\mu_T)\dto Z,\qquad
 \sqrt{T}m_T\dto h+Z.
\]
The covariance limit is positive semidefinite and, because the observations are unit vectors and $\mu_T\to0$,
\[
 \tr\Sigma=\lim_{T\to\infty}\tr\Sigma_T
 =\lim_{T\to\infty}\bigl(1-\norm{\mu_T}^2\bigr)=1.
\]
Thus $\Sigma\ne0$ automatically. Some one-dimensional projection of $Z$ has positive variance and a continuous Gaussian law, which implies $\Prob(h+Z=0)=0$ even when $\Sigma$ is singular. Continuous mapping gives the joint limit. The same absence of an atom at zero gives $\Prob(m_T=0)\to0$, making the arbitrary convention immaterial.
\end{proof}

In this regime the city is not stable, and for $h\ne0$ it is not uniformly random either: it is the direction of a fixed displacement $h$ plus Gaussian noise of comparable size. The small retained lengths in the competition sample (median $0.0338$) make this regime relevant there, without establishing that it applies.

For $h=0$ and $\Sigma=\id_q/q$, the limiting direction is uniform. Singular covariance limits are also allowed: if $h=0$ and $\Sigma=ww^\top$ for a unit vector $w$, the limiting direction takes the two values $w$ and $-w$ with equal probability. For nonzero finite $h$, the limit retains both a mean displacement and random fluctuations, while $r_T$ remains of order $T^{-1/2}$ in probability. The proposition describes distributions at the transition scale; a model for joint motion across dates requires additional temporal assumptions.

\Needspace{17\baselineskip}
\section{Numerical checks and competition evidence}\label{sec:empirical}

\subsection{Synthetic checks}

The preceding claims are analytic. We report three synthetic checks of their formulas, using random seed 260926 for the random samples. \Cref{ex:finite} gives the correlation $0.985$ and datewise IC $0.05$ directly. For a six-date, four-dimensional pair of normalized Gaussian vector samples, all $6!=720$ permutations were enumerated. The results were
\begin{center}
\small
\begin{tabular}{@{}lrr@{}}
\toprule
Quantity & Enumeration & Formula\\
\midrule
Mean & $-0.045388$ & $-0.045388$\\
Variance & $0.051354$ & $0.051354$\\
\bottomrule
\end{tabular}
\end{center}
The enumerated moments and formulas agree to machine precision; the table shows six decimal places. The nested intervals were also checked on 1,000 synthetic 24-date pairs in dimension five, using partitions into 1, 2, 4, 8, and 24 equal contiguous blocks. No numerical violation was observed. These checks verify numerical implementations of the identities.

\subsection{Competition May 2026: sample and measurement}\label{sec:competition}

We study an export of the prediction histories submitted to \href{https://www.alphanova.tech}{AlphaNova}'s Competition May 2026. The histories are stored in gauge space, that is, in the common target-aligned frame of \cref{sec:setup} under the export's stored frame convention. ``May 2026'' identifies when the signals were submitted; the histories themselves span a multi-year historical evaluation. The study is therefore retrospective.

The export contains 919 histories with 19 coordinates per row ($q=19$, corresponding to 20 demeaned assets). File hashes and numeric-content hashes identify the same 843 distinct histories, with 76 additional copies in 17 groups. Nine distinct histories cover fewer than 77 of the 80 reference month blocks and are excluded from the main comparison. The remaining 834 histories share 56,120 stored rows. Blocks are aligned by their starting row labels, with observations keyed by block and row label. Calendar provenance and repeated boundary observations remain subject to the limitations below.

The main sample, denoted U, enforces the unit-vector condition. A zero row has norm below $10^{-12}$; a unit row has norm within $10^{-9}$ of one. There are 400 shared non-unit rows across the full reference layout, five per month; 378 lie on the main common support. We build U in three steps. First, we drop those 378 rows. Second, we exclude 58 signals that have more than $1\%$ zero rows on that support. Third, we drop every row position at which any retained signal has a zero or non-unit row: 3,907 zero-row positions and two other non-unit positions. The result is 776 distinct histories on 51,833 common rows in 76 month blocks. These data-quality rules were fixed before the separation statistics were computed.

The 378-row count reflects incomplete month ends, rather than 76 complete sets of five. The shared non-unit rows are the last five positions of each reference block. Months 0--74 contribute $75\times5=375$ rows. In month 75, 19 histories stop at row 724 of 726, removing two of those five positions from the common support; three remain. The support also includes only the first 178 rows of month 76, before its non-unit month-end positions, so that block contributes none. Thus the count is $375+3+0=378$. After all row exclusions, U has observations only in months 0--75.

For each of the $\binom{776}{2}=300{,}700$ unordered pairs, means, retained lengths, and cities are recomputed on the same support. We calculate $D_{ij}$, $C_{ij}$, and
\[
 K_{ij}=D_{ij}-r_ir_jC_{ij},\qquad
 \Delta_{ij}=\arccos D_{ij}-\arccos C_{ij}.
\]
Separation is \emph{preserved} when $D_{ij}\le C_{ij}$, equivalently $\Delta_{ij}\ge0$: the city angle is then a conservative summary of the full-history angle. It is \emph{overstated} when $D_{ij}>C_{ij}$. Numerical ties within $10^{-10}$ are counted as preserved; the 17 ties are pairs identical on the retained support. All comparisons are signed. Each compares one pair's city angle with the same pair's history angle; we do not compare a signal's nearest neighbor among cities with its possibly different nearest neighbor among histories.

\Needspace{11\baselineskip}
\subsection{Positive-IC signals and conservative city separation}\label{sec:positiveic}

Of the 776 histories in U, 731 have strictly positive average IC on the common evaluation support. Their 266,815 unordered pairs, in each of which both signals have positive average IC, form the \emph{positive-IC subset} used for the following descriptive comparisons; we also call them \emph{both-positive} pairs. History separation is at least city separation for $157{,}198/266{,}815=58.9\%$ of these pairs. Among the 210,951 pairs whose city angle is at most $90^\circ$, the fraction is $156{,}362/210{,}951=74.1\%$. Beyond $90^\circ$, it is $836/55{,}864=1.5\%$.

In this subset, then, the city angle is a conservative estimate of the history angle for a majority of pairs, and for about three quarters of the pairs whose city angle is at most $90^\circ$. \Cref{tab:empiricalrates} compares the positive-IC subset, its complement, and the full sample. The $74.1\%$ figure is an average over city angles from $0^\circ$ to $90^\circ$; the rate is not the same at every angle, as \cref{fig:empiricalbins} shows. The table's northern null row gives reference limits for the both-positive selection only. It is not a calibrated prediction for this library, and it is not the appropriate selection law for the complement.

The positive-IC subset dominates the full sample, so comparing those two hides how different the complement is. The complement consists of the 33,885 pairs that involve at least one of the 45 signals without strictly positive average IC. Subtracting the positive-IC counts from the full-sample counts gives 19,660 complement pairs with $C\ge0$, of which 18,198 preserve separation, and 14,225 with $C<0$, of which 531 preserve it. Three contrasts follow. First, $58.0\%$ of complement pairs have nonnegative city correlation, against $79.1\%$ of both-positive pairs. Second, for $C\ge0$, preservation is $92.6\%$ in the complement and $74.1\%$ among both-positive pairs. Third, for $C<0$, it is $3.7\%$ and $1.5\%$, respectively. Overall, the complement preserves separation in $18{,}729/33{,}885=55.3\%$ of cases. In this sample, therefore, the sign of average IC is associated with different city geometry and different conditional preservation. These comparisons are descriptive; they do not identify a causal effect of selecting for positive IC.

\begin{table}[!htb]
\centering\small
\caption{Preservation by empirical IC-sign group, with $T=51{,}833$. Conditional columns report the frequency of $D\le C$ within the indicated city-correlation regime. The complement consists of pairs with at least one nonpositive mean IC. The northern null row applies only to the both-positive selection, with $q=19$ and $T\to\infty$; it is not a benchmark for the complement. Its limits also hold under \cref{thm:dependentordering} when $a_T\to0$.}
\label{tab:empiricalrates}
\begin{tabular}{@{}lrrrrr@{}}
\toprule
Pair set & Pairs & \shortstack{Share with\\$C\ge0$} & \shortstack{Preserved\\given $C\ge0$} & \shortstack{Preserved\\given $C<0$} & \shortstack{Preserved\\overall}\\
\midrule
Both mean ICs positive & 266,815 & $79.1\%$ & $74.1\%$ & $1.5\%$ & $58.9\%$\\
\shortstack[l]{At least one mean IC\\not positive} & 33,885 & $58.0\%$ & $92.6\%$ & $3.7\%$ & $55.3\%$\\
All pairs in U & 300,700 & $76.7\%$ & $75.7\%$ & $2.0\%$ & $58.5\%$\\
\midrule
\shortstack[l]{Northern null limit\\(both positive only)} & --- & $56.0\%$ & $100\%$ & $0\%$ & $56.0\%$\\
\bottomrule
\end{tabular}
\end{table}

In the full sample U, including signals without positive average IC, preservation occurs for $175{,}927/300{,}700=58.5\%$ of all pairs, $174{,}560/230{,}611=75.7\%$ of nonnegative-city pairs, and $1{,}367/70{,}089=2.0\%$ of negative-city pairs. These are pair-weighted rates for the ordering $D\le C$. The separate threshold question is evaluated in \cref{sec:screening}. In the remaining 124,773 pairs ($41.5\%$), the city angle overstates the history angle. The median shortfall is $11.9^\circ$, and for one quarter of these pairs it exceeds $21.7^\circ$. Both how often preservation fails and by how much therefore matter.

\Needspace{9\baselineskip}
Among the 230,611 pairs with $C_{ij}\ge0$, the exact decomposition is
\begin{center}
\small
\begin{tabular}{@{}lrr@{}}
\toprule
Region & Pairs & Share of $C_{ij}\ge0$\\
\midrule
$K_{ij}\le0$ & 19,285 & $8.4\%$\\
$0<K_{ij}\le(1-r_ir_j)C_{ij}$ & 155,275 & $67.3\%$\\
$K_{ij}>(1-r_ir_j)C_{ij}$ & 56,051 & $24.3\%$\\
\bottomrule
\end{tabular}
\end{center}
Most preserved pairs therefore lie in the admissible positive-residual band, not in the sufficient region $K_{ij}\le0$. For both-positive pairs the three shares are $6.9\%$, $67.2\%$, and $25.9\%$. In the complement, subtraction gives 4,667 nonnegative-city pairs preserved with $K\le0$, 13,531 preserved with $0<K\le(1-r_ir_j)C$, and 1,462 violations, or $23.7\%$, $68.8\%$, and $7.4\%$, respectively. Conditional preservation is thus 18.4 percentage points higher in the complement than among both-positive pairs. Of that gap, 16.8 points come from the complement's larger share of nonpositive residuals and only 1.6 points from the positive-residual band. This is an arithmetic decomposition of the observed gap, not a statement about separate causal effects.

The same contrast in residual signs appears when all pairs in each group are counted, not only those with $C\ge0$. There are 33,961 full-sample pairs with $K\le0$, including 23,919 both-positive pairs. The complement therefore has $10{,}042/33{,}885=29.6\%$ with $K\le0$, against $9.0\%$ in the both-positive subset. Equivalently, the residual cross-moment is positive for $23{,}843/33{,}885=70.4\%$ of complement pairs and $242{,}896/266{,}815=91.0\%$ of both-positive pairs. In this sample, positive residual agreement is thus more frequent among positive-IC pairs. Residual signs alone neither show that the residuals are larger in magnitude nor identify a common-factor mechanism.

Retained lengths are small: the median is $r_i=0.0338$, the 95th percentile $0.1177$, and the maximum $0.2567$. Hence $|D_{ij}-K_{ij}|=|r_ir_jC_{ij}|\le0.0659$ for every pair in U, so the residual is an important part of the observed history agreement. Because the retained lengths are small, the full-history certificate in \cref{cor:screen} certifies no pair at threshold $0.5$: even at the maximum length, its signed upper endpoint is at least $1-2(0.2567)^2\simeq0.8682$. By the sharpness part of \cref{thm:sharp}, this is a limit on the information in a pair's cities and retained lengths, not a slack inequality. \Cref{sec:blockexperiment} applies the block refinement of \cref{thm:blocks} to the same pairs. Small retained lengths alone, however, neither force $D_{ij}$ near zero nor determine the sign of $D_{ij}-C_{ij}$. The exact criterion remains $K_{ij}\le(1-r_ir_j)C_{ij}$, including when $C_{ij}<0$.

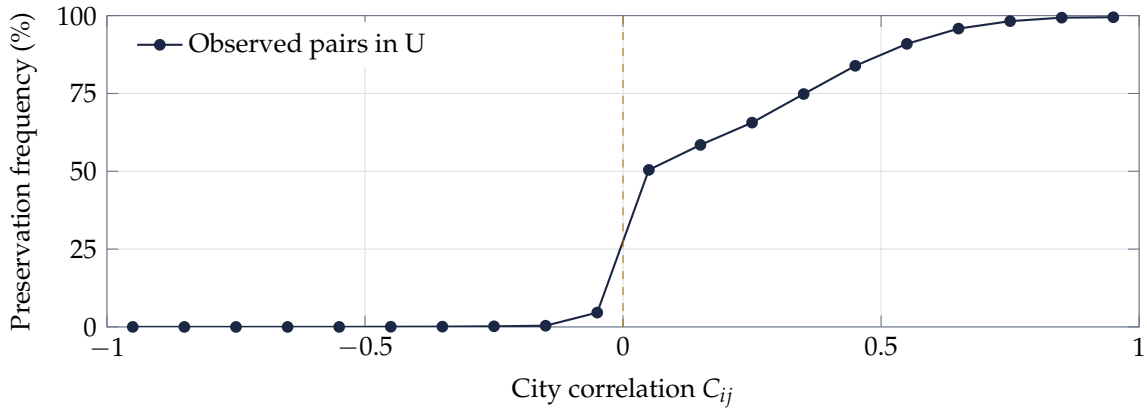
\begin{figure}[!htb]
\centering
\begin{tikzpicture}
\begin{axis}[
 width=.89\textwidth,height=5.7cm,
 xmin=-1,xmax=1,ymin=0,ymax=100,
 xtick={-1,-.5,0,.5,1},ytick={0,25,50,75,100},
 xlabel={City correlation $C_{ij}$},
 ylabel={Preservation frequency (\%)},
 tick label style={font=\small},label style={font=\small},
 axis line style={muted},tick style={muted},
 grid=major,grid style={rulegray},
 legend style={draw=none,font=\small,at={(.02,.98)},anchor=north west}
]
\addplot[navy,thick,mark=*,mark size=1.8pt] coordinates {
(-0.95,0.00000000)
(-0.85,0.00000000)
(-0.75,0.00000000)
(-0.65,0.00000000)
(-0.55,0.00000000)
(-0.45,0.05543237)
(-0.35,0.09504198)
(-0.25,0.19750520)
(-0.15,0.38951841)
(-0.05,4.60526316)
(0.05,50.44449093)
(0.15,58.46482349)
(0.25,65.62489752)
(0.35,74.81247694)
(0.45,83.87875536)
(0.55,90.92527864)
(0.65,95.82815306)
(0.75,98.22238568)
(0.85,99.34014474)
(0.95,99.47855994)
};
\addlegendentry{Observed pairs in U}
\addplot[gold,dashed,forget plot] coordinates {(0,0) (0,100)};
\end{axis}
\end{tikzpicture}
\caption{Observed preservation by city correlation, in bins of width $0.1$ (points at bin centers). Frequencies rise from $50.4\%$ for $0\le C_{ij}<0.1$ to above $99\%$ for $C_{ij}>0.8$. Bin sizes range from 370 to 38,246 pairs. The highest rates occur for closely aligned cities. Because the threshold $(1-r_ir_j)C_{ij}$ itself increases approximately with $C_{ij}$ when retained lengths are small, this profile alone does not establish the predictive information in city alignment; \cref{sec:screening} supplies the threshold-specific comparison with base rates.}
\label{fig:empiricalbins}
\end{figure}

\subsection{Threshold screening}\label{sec:screening}

Fix a history-correlation target $\tau$ and a city-correlation cutoff $c$. The screen accepts a pair when $C_{ij}\le c$; success means $D_{ij}\le\tau$. With matched cutoffs $c=\tau$, this asks whether a city angle of at least $\arccos\tau$ identifies a history angle of at least the same size. Screening results below use all 776 signals in U, without a positive-IC restriction. All probabilities in this empirical section denote fractions of the observed pairs. We report the conditional success rate and its improvement over the base rate,
\[
 \Prob(D\le\tau\mid C\le c),\qquad
 \Prob(D\le\tau\mid C\le c)-\Prob(D\le\tau),
\]
together with pair retention, $\Prob(C\le c)$. For an accepted pair, threshold success asks less than the preservation event $D\le C$: the history angle must clear the threshold, but it need not be as large as the city angle. Thus, when $C\le\tau$, preservation is sufficient for $D\le\tau$ but not necessary.

\begin{table}[!htb]
\centering\small
\caption{Matched-threshold screening on U: accept $C\le\tau$ and require $D\le\tau$. Success is conditional on acceptance; the base rate uses all 300,700 pairs. Improvement is in percentage points (pp).}
\label{tab:screening}
\begin{tabular}{@{}rrrrrr@{}}
\toprule
$\tau$ & Accepted pairs & Retained & Success & Base rate & Improvement\\
\midrule
$0.75$ & 255,790 & $85.1\%$ & $99.5\%$ & $97.7\%$ & $+1.8$ pp\\
$0.50$ & 217,926 & $72.5\%$ & $94.3\%$ & $85.0\%$ & $+9.3$ pp\\
$0.25$ & 160,414 & $53.3\%$ & $75.6\%$ & $56.2\%$ & $+19.4$ pp\\
$0$ & 70,089 & $23.3\%$ & $21.2\%$ & $11.2\%$ & $+10.0$ pp\\
$-0.50$ & 5,940 & $2.0\%$ & $3.7\%$ & $0.073\%$ & $+3.6$ pp\\
\bottomrule
\end{tabular}
\end{table}

At the primary threshold $\tau=0.5$, corresponding to $60^\circ$, the screen accepts 217,926 pairs, of which 205,490 meet the history threshold. The $94.3\%$ success rate exceeds the $85.0\%$ base rate by $9.3$ percentage points. Put differently, the proportion of pairs with $D>0.5$ falls from $15.0\%$ before screening to $5.7\%$ among accepted pairs, a relative reduction of $62.0\%$. These rates are conditional on acceptance. Conditioning instead on the history outcome gives a different picture: $27.5\%$ of all pairs with $D>0.5$ still pass the screen, and $19.6\%$ of all pairs with $D\le0.5$ are rejected, so the screen retains $80.4\%$ of the acceptable pairs. Among rejected pairs, $32{,}715/82{,}774=39.5\%$ exceed the history-correlation limit and $60.5\%$ meet it. \Cref{tab:screeningcounts} gives the complete counts.

\begin{table}[!htb]
\centering\small
\caption{Pair counts at the primary threshold. The 94.3\% success rate is $205{,}490/217{,}926$; the 85.0\% base rate is $255{,}549/300{,}700$.}
\label{tab:screeningcounts}
\begin{tabular}{@{}lrrr@{}}
\toprule
City decision & $D\le0.5$ & $D>0.5$ & Total\\
\midrule
Accepted: $C\le0.5$ & 205,490 & 12,436 & 217,926\\
Rejected: $C>0.5$ & 50,059 & 32,715 & 82,774\\
\midrule
Total & 255,549 & 45,151 & 300,700\\
\bottomrule
\end{tabular}
\end{table}

Each row of \cref{tab:screening} changes the city cutoff and the history target together. At $\tau=0.75$, success is $99.5\%$, but the base rate is already $97.7\%$. At $\tau=0.25$, the improvement is $19.4$ percentage points, the largest among the reported thresholds. At $\tau=0$ and $-0.5$, success exceeds the base rate but remains only $21.2\%$ and $3.7\%$, respectively: opposite city directions do not usually imply nonpositive or equally negative history correlation. At $\tau=-0.5$, all 220 pairs with $D\le-0.5$ pass the screen, so the exact success and base fractions are $220/5{,}940$ and $220/300{,}700$, respectively. These comparisons are descriptive and do not identify a universally optimal threshold.

To isolate the effect of tightening the city screen, \cref{tab:fixedtarget} holds the history target at $D\le0.5$. Moving from $c=0.5$ to $c=0$ increases success by only $2.5$ percentage points, from $94.3\%$ to $96.8\%$, while pair retention falls from $72.5\%$ to $23.3\%$ and acceptable-pair retention falls from $80.4\%$ to $26.5\%$. In this sample, the tighter cutoff buys a small gain in success rate at the cost of most acceptable pairs.

\begin{table}[!htb]
\centering\small
\caption{A fixed history target, $D\le0.5$, with varying city cutoff. Acceptable-pair retention is $\Prob(C\le c\mid D\le0.5)$. The no-screen baseline includes all 300,700 pairs.}
\label{tab:fixedtarget}
\begin{tabular}{@{}lrrr@{}}
\toprule
City cutoff $c$ & Pairs retained & Success & Acceptable pairs retained\\
\midrule
No screen & $100.0\%$ & $85.0\%$ & $100.0\%$\\
$0.75$ & $85.1\%$ & $92.4\%$ & $92.5\%$\\
$0.50$ & $72.5\%$ & $94.3\%$ & $80.4\%$\\
$0.25$ & $53.3\%$ & $95.8\%$ & $60.1\%$\\
$0$ & $23.3\%$ & $96.8\%$ & $26.5\%$\\
$-0.50$ & $2.0\%$ & $96.4\%$ & $2.2\%$\\
\bottomrule
\end{tabular}
\end{table}

\subsection{History correlations and screening uncertainty}\label{sec:screenuncertainty}

The distribution of $D$ explains the high base rate at the primary threshold. Its median is $0.203$, mean $0.238$, and interquartile range $[0.021,0.401]$; its 5th and 95th percentiles are $-0.012$ and $0.665$. The fractions with $D\le0$, $0.25$, $0.5$, and $0.75$ are $11.2\%$, $56.2\%$, $85.0\%$, and $97.7\%$, respectively. City and history correlation are associated: their Spearman rank correlation is $0.561$, with a month-block bootstrap median of $0.556$ and a 95\% percentile range of $[0.429,0.649]$. Thus $C_{ij}$ can be informative about $D_{ij}$ even though the explicit mean term $r_ir_jC_{ij}$ contributes little to it: city alignment and residual co-movement need not be independent.

Pairs share signals and signals share scientists, so the 300,700 pair comparisons cannot be treated as independent trials. To indicate variability over time, we resample the 76 retained month blocks with replacement, a nonoverlapping block bootstrap~\cite{carlstein,lahiri}, apply the same resampled months to every signal, and recompute means, cities, $C$, and $D$ in each replicate. The screening bootstrap uses 1,000 replicates and seed 20261002. \Cref{tab:screenbootstrap} reports the original estimates, bootstrap medians, and 2.5--97.5 percentile ranges.

\begin{table}[!htb]
\centering\small
\caption{Month-block resampling on U. Screening rows use $c=\tau=0.5$; improvement is in percentage points. Angle preservation uses $D\le C$ over all pairs.}
\label{tab:screenbootstrap}
\begin{tabular}{@{}lrrr@{}}
\toprule
Statistic & Original & Bootstrap median & 95\% range\\
\midrule
Pairs accepted & $72.5\%$ & $74.2\%$ & $70.4$--$76.6\%$\\
Success among accepted & $94.3\%$ & $94.1\%$ & $92.7$--$95.2\%$\\
Base rate & $85.0\%$ & $84.9\%$ & $83.3$--$86.6\%$\\
Improvement & $9.3$ pp & $9.1$ pp & $8.2$--$10.0$ pp\\
Accepted given $D>0.5$ & $27.5\%$ & $29.1\%$ & $24.8$--$33.4\%$\\
Rejected given $D\le0.5$ & $19.6\%$ & $17.8\%$ & $15.5$--$22.3\%$\\
\midrule
Angle preservation & $58.5\%$ & $54.2\%$ & $47.5$--$61.6\%$\\
\bottomrule
\end{tabular}
\end{table}

Not every bootstrap distribution is centered on its original estimate. At the primary cutoff, the bootstrap median of pair acceptance is $1.7$ percentage points above the original value, and that of angle preservation is $4.3$ points below it ($54.2\%$ against $58.5\%$). Resampling repeats some months and omits others; because the mean vectors are small, the resulting variation can materially change city directions and the distribution of $C$. This is consistent with the observed shift but does not, taken alone, isolate its cause. Screening success and improvement at the primary cutoff each shift by only $-0.2$ points. The improvement range of $8.2$--$10.0$ points is therefore informative about this resampling scheme; it is not a guarantee for future signals or dates. Resampling months independently preserves shared movement within each selected block. It does not preserve dependence between successive months, and it does not jointly model uncertainty across scientists and time; \cref{sec:robustness} reports a two-way resampling.

The same month draws give ranges for the two IC-sign groups of \cref{sec:positiveic}, with membership fixed at the 731 signals that have positive average IC on the original support (\cref{tab:icsignbootstrap}). For both-positive pairs, overall preservation has a bootstrap median of $54.6\%$ and a range of $48.3$--$61.2\%$, which extends below one half. For city angles of at most $90^\circ$ the range is $68.2$--$75.7\%$. The gap in that regime between the complement and the both-positive pairs, $18.4$ percentage points in the original sample, has a bootstrap median of $20.0$ points and a range of $15.9$--$23.9$, so its sign is stable under this resampling. Recomputing positive-IC membership within each replicate changes little: the number of positive-IC signals has median 726 (range 694--753), and the both-positive ranges become $48.6$--$61.5\%$ overall and $67.7$--$75.8\%$ for $C\ge0$. The complement's ranges are wide because it rests on 45 signals.

\begin{table}[!htb]
\centering\small
\caption{Month-block resampling by IC-sign group: the 1,000 month draws of \cref{tab:screenbootstrap}, with positive-IC membership fixed at the original 731 signals. Each cell gives the original value, the bootstrap median, and the 2.5--97.5 percentile range, in percent.}
\label{tab:icsignbootstrap}
\begin{tabular}{@{}lrrrrrr@{}}
\toprule
 & \multicolumn{3}{c}{Both-positive pairs} & \multicolumn{3}{c}{Complement}\\
\cmidrule(lr){2-4}\cmidrule(l){5-7}
Statistic & Original & Median & 95\% range & Original & Median & 95\% range\\
\midrule
Preserved & $58.9$ & $54.6$ & $48.3$--$61.2$ & $55.3$ & $50.8$ & $37.4$--$66.4$\\
Preserved, $C\ge0$ & $74.1$ & $71.8$ & $68.2$--$75.7$ & $92.6$ & $92.2$ & $87.2$--$95.3$\\
Preserved, $C<0$ & $1.5$ & $1.0$ & $0.5$--$2.6$ & $3.7$ & $2.5$ & $1.3$--$5.1$\\
Share with $C\ge0$ & $79.1$ & $75.7$ & $69.1$--$81.6$ & $58.0$ & $54.0$ & $40.5$--$69.4$\\
Share with $K\le0$ & $9.0$ & $9.8$ & $5.7$--$15.7$ & $29.6$ & $32.3$ & $16.5$--$53.4$\\
\bottomrule
\end{tabular}
\end{table}

\subsection{Block means: how fine a partition the certificates need}\label{sec:blockexperiment}

Because the single-mean certificate decides no pair in this library (\cref{sec:positiveic}), we apply the block intervals of \cref{thm:blocks} to U. For a partition $\mathcal P$ of the 51,833 rows, each signal stores one mean vector per block, and the interval $I_{\mathcal P}$ of \cref{eq:blockinterval} has half-width $\sqrt{E_i^{\mathcal P}E_j^{\mathcal P}}$. We use two nested families. The \emph{dyadic} family splits the stored row order into $2^k$ consecutive blocks of nearly equal length, for $k=0,\ldots,15$. The \emph{calendar-like} family uses the 76 month blocks, then consecutive chunks of 168, 24, and 6 rows inside each month. Both end at single rows, where the interval is the point $D_{ij}$. Chunks are counted in stored rows; without timestamps they are not verified weeks, days, or hours. For each of the 21 distinct partitions (the one-block and single-row partitions belong to both families) and all 300,700 pairs, $D_{ij}$ lies in its interval and every finer interval lies inside the coarser one, with no violation at tolerance $10^{-10}$.

A pair is \emph{certified} at threshold $\tau$ when the upper endpoint of its interval is at most $\tau$, as in \cref{cor:screen}. \Cref{tab:blockcertificates} reports, among the pairs with $D_{ij}\le\tau$, the share that is certified.

\begin{table}[!htb]
\centering\small
\caption{Block-mean certificates on U. $B$ is the number of blocks, and $B/T$ the number of stored mean vectors per signal relative to its 51,833 rows. The last three columns give the percentage of pairs with $D_{ij}\le\tau$ whose interval certifies $D_{ij}\le\tau$.}
\label{tab:blockcertificates}
\begin{tabular}{@{}lrrrrrr@{}}
\toprule
Partition & $B$ & $B/T$ & Median half-width & $\tau=0.75$ & $\tau=0.5$ & $\tau=0.25$\\
\midrule
\multicolumn{7}{@{}l}{\emph{Dyadic family}}\\
One block (the city) & $1$ & $0.0019\%$ & $0.998$ & $0.0$ & $0.0$ & $0.0$\\
64 blocks & $64$ & $0.12\%$ & $0.958$ & $0.1$ & $0.0$ & $0.0$\\
256 blocks & $256$ & $0.49\%$ & $0.914$ & $1.7$ & $0.0$ & $0.0$\\
1{,}024 blocks & $1{,}024$ & $2.0\%$ & $0.829$ & $14.4$ & $0.3$ & $0.0$\\
2{,}048 blocks & $2{,}048$ & $4.0\%$ & $0.759$ & $33.0$ & $1.6$ & $0.0$\\
4{,}096 blocks & $4{,}096$ & $7.9\%$ & $0.663$ & $55.7$ & $8.0$ & $0.2$\\
8{,}192 blocks & $8{,}192$ & $16\%$ & $0.535$ & $74.3$ & $28.5$ & $1.5$\\
16{,}384 blocks & $16{,}384$ & $32\%$ & $0.380$ & $89.6$ & $54.7$ & $12.5$\\
32{,}768 blocks & $32{,}768$ & $63\%$ & $0.186$ & $97.6$ & $83.4$ & $59.9$\\
\midrule
\multicolumn{7}{@{}l}{\emph{Calendar-like family}}\\
Month blocks & $76$ & $0.15\%$ & $0.943$ & $0.5$ & $0.0$ & $0.0$\\
168-row chunks & $353$ & $0.68\%$ & $0.896$ & $3.4$ & $0.0$ & $0.0$\\
24-row chunks & $2{,}189$ & $4.2\%$ & $0.751$ & $35.3$ & $1.9$ & $0.0$\\
6-row chunks & $8{,}673$ & $17\%$ & $0.522$ & $75.8$ & $30.9$ & $1.8$\\
\midrule
Single rows & $51{,}833$ & $100\%$ & $0.000$ & $100.0$ & $100.0$ & $100.0$\\
\bottomrule
\end{tabular}
\end{table}

The certificates need fine resolution. At $\tau=0.5$, no pair is certified with 64 blocks or fewer, and fewer than $0.05\%$ of the 255,549 pairs with $D\le0.5$ are certified with 512 blocks or fewer, with month blocks, or with 168-row chunks. The share reaches $28.5\%$ at 8,192 blocks, where each signal stores $16\%$ as many vectors as its history has rows; $54.7\%$ at 16,384 blocks ($32\%$); and $83.4\%$ at 32,768 blocks ($63\%$). In the calendar-like family, 24-row chunks certify $1.9\%$ and 6-row chunks $30.9\%$, with a month-block bootstrap range of $29.2$--$32.5\%$ (200 replicates). A looser threshold is cheaper: at $\tau=0.75$, 4,096 blocks ($7.9\%$ of the rows) certify $55.7\%$ of the pairs below the threshold, and 24-row chunks certify $35.3\%$. A tighter threshold is dearer: at $\tau=0.25$, only $12.5\%$ are certified at 16,384 blocks.

The residual energies show why. With one block the median $E_i$ is $0.999$; with month blocks it is still $0.954$, and with 24-row chunks $0.769$. Most of each history's energy therefore stays in the unobserved residual until blocks are a few rows long. This is consistent with the lag profiles of \cref{sec:theorydata}, whose median lag-one value is $0.469$: the block means of signals that change this quickly remain short.

Results for the absolute rule and for both-positive pairs are similar. At 16,384 blocks and $\tau=0.5$, the absolute certificate covers $54.2\%$ of the pairs with $|D|\le0.5$, and the signed certificate $49.9\%$ of the both-positive pairs with $D\le0.5$. The lower endpoint certifies $D_{ij}>\tau$ for few pairs before full resolution: $1.5\%$ of all pairs at 32,768 blocks and $\tau=0.5$.

\paragraph{Block means as a screen.}
The block means also define a compressed history whose correlation can serve as a point estimate. Put
\[
 C^{\mathcal P}_{ij}=\frac{A^{\mathcal P}_{ij}}{\sqrt{(1-E^{\mathcal P}_i)(1-E^{\mathcal P}_j)}},
\]
which is the city correlation $C_{ij}$ for one block and $D_{ij}$ for single rows. \Cref{tab:blockscreening} repeats the screening study of \cref{sec:screening} at $c=\tau=0.5$ with $C^{\mathcal P}$ in place of $C$. This is a descriptive screen, not a certificate.

\begin{table}[!htb]
\centering\small
\caption{Screening on U with the block correlation $C^{\mathcal P}$ at $c=\tau=0.5$; the base rate is $85.0\%$ throughout. Entries other than the rank correlation are percentages. Gain is success minus the base rate, in percentage points. The last two columns condition on the history outcome: the share accepted among pairs with $D>0.5$, and the share rejected among pairs with $D\le0.5$.}
\label{tab:blockscreening}
\setlength{\tabcolsep}{4.5pt}
\begin{tabular}{@{}lrrrrrrrr@{}}
\toprule
 & & Spearman & & & & & Accepted, & Rejected,\\
Partition & $B$ & with $D$ & $D\le C^{\mathcal P}$ & Accepted & Success & Gain & $D>0.5$ & $D\le0.5$\\
\midrule
\multicolumn{9}{@{}l}{\emph{Calendar-like family}}\\
One block & $1$ & $0.561$ & $58.5$ & $72.5$ & $94.3$ & $+9.3$ & $27.5$ & $19.6$\\
Month blocks & $76$ & $0.666$ & $36.8$ & $83.5$ & $92.8$ & $+7.9$ & $39.8$ & $8.8$\\
168-row chunks & $353$ & $0.752$ & $34.8$ & $84.6$ & $93.1$ & $+8.1$ & $38.9$ & $7.3$\\
24-row chunks & $2{,}189$ & $0.853$ & $37.4$ & $85.4$ & $94.2$ & $+9.2$ & $33.0$ & $5.3$\\
6-row chunks & $8{,}673$ & $0.938$ & $41.8$ & $84.6$ & $96.3$ & $+11.3$ & $20.8$ & $4.1$\\
\midrule
\multicolumn{9}{@{}l}{\emph{Dyadic family}}\\
16 blocks & $16$ & $0.655$ & $42.2$ & $81.0$ & $93.4$ & $+8.4$ & $35.7$ & $11.0$\\
256 blocks & $256$ & $0.731$ & $34.8$ & $84.6$ & $93.0$ & $+8.0$ & $39.6$ & $7.5$\\
4{,}096 blocks & $4{,}096$ & $0.897$ & $38.8$ & $85.1$ & $95.1$ & $+10.1$ & $27.8$ & $4.8$\\
16{,}384 blocks & $16{,}384$ & $0.965$ & $44.9$ & $84.5$ & $97.3$ & $+12.3$ & $15.1$ & $3.2$\\
32{,}768 blocks & $32{,}768$ & $0.991$ & $47.6$ & $84.8$ & $98.9$ & $+13.9$ & $6.4$ & $1.4$\\
\midrule
Single rows & $51{,}833$ & $1.000$ & $100.0$ & $85.0$ & $100.0$ & $+15.0$ & $0.0$ & $0.0$\\
\bottomrule
\end{tabular}
\end{table}

The rank correlation between $C^{\mathcal P}$ and $D$ rises with resolution, from $0.561$ for one block to $0.666$ for month blocks and $0.938$ for 6-row chunks. Screening success does not rise monotonically. Month blocks accept more pairs than the single city ($83.5\%$ against $72.5\%$), including more of the pairs with $D>0.5$ ($39.8\%$ against $27.5\%$). Success among accepted pairs therefore falls from $94.3\%$ to $92.8\%$, and the improvement over the base rate from $9.3$ to $7.9$ percentage points. What month blocks gain is retention: they reject $8.8\%$ of the pairs with $D\le0.5$, against $19.6\%$. The improvement returns to $9.2$ points with 24-row chunks and reaches $11.3$ with 6-row chunks (month-block bootstrap range $10.5$--$12.2$). The share of pairs with $D\le C^{\mathcal P}$ is not monotone either: it falls from $58.5\%$ to about $35\%$ at intermediate resolutions before reaching $100\%$ at single rows.

An exact identity reads this second pattern. Let $K^{\mathcal P}_{ij}$ be the average inner product of the two within-block residual histories, so that $D_{ij}=A^{\mathcal P}_{ij}+K^{\mathcal P}_{ij}$ by the orthogonal decomposition in the proof of \cref{thm:blocks}, and, when both residual energies are positive, let $\rho^{\mathcal P}_{ij}=K^{\mathcal P}_{ij}/\sqrt{E^{\mathcal P}_iE^{\mathcal P}_j}$ be the correlation of those residual histories. Then
\[
 D_{ij}=\sqrt{(1-E^{\mathcal P}_i)(1-E^{\mathcal P}_j)}\;C^{\mathcal P}_{ij}
 +\sqrt{E^{\mathcal P}_iE^{\mathcal P}_j}\;\rho^{\mathcal P}_{ij},
\]
and $D_{ij}\le C^{\mathcal P}_{ij}$ exactly when $K^{\mathcal P}_{ij}\le\bigl(1-\sqrt{(1-E^{\mathcal P}_i)(1-E^{\mathcal P}_j)}\bigr)C^{\mathcal P}_{ij}$, the block analogue of \cref{eq:ordering}. When the two residual energies are equal, the criterion is simply $\rho^{\mathcal P}_{ij}\le C^{\mathcal P}_{ij}$. The fall to about $35\%$ therefore says that, for most pairs at intermediate resolutions, the residual agreement exceeds its margin; with comparable residual energies, the two signals' within-block residuals are more correlated than their block means. The shared movement between signals lies more within blocks than between them, which is consistent with the residual energies reported above. The cutoff was not optimized, and the comparison is retrospective on the same support.

\subsection{Interpreting the theory--data comparison}\label{sec:theorydata}

The empirical residual distribution explains which pairs preserve separation. In U, $88.7\%$ of residual cross-moments are positive; the share is $91.0\%$ in the both-positive-IC subset and $70.4\%$ in its complement (\cref{sec:positiveic}). The exact sample identity
\[
 D_{ij}-C_{ij}=K_{ij}-(1-r_ir_j)C_{ij}
\]
locates each violation where residual agreement exceeds the available margin. It requires no stochastic sampling assumptions. Preservation with $C<0$ requires $D<0$, and only $11.2\%$ of full-sample pairs have $D\le0$, so such pairs can come only from a small part of the library. For that reason the qualitative contrast across $90^\circ$ is, taken alone, only a weak test of the uniform null.

The overall $58.9\%$ rate in the positive-IC subset is close to the northern null limit $\gamma_{19}\simeq56.0\%$, but the agreement is the net result of two differences that offset each other. In that subset, $79.1\%$ of pairs have $C\ge0$, against $56.0\%$ under the northern null. Within that regime, preservation is $74.1\%$, against a limit of one under the null with $a_T\to0$. Residual signs differ as well: the share of nonpositive residuals is $9.0\%$, whereas the independent-date null has a limiting share of one half. (\Cref{thm:dependentordering} does not assert that residual-sign limit under arbitrary serial dependence.)

\paragraph{Finite-sample dependence sensitivity.}
The five-hour forward-return target overlaps at adjacent hourly observations, and predictions can remain dependent over longer periods. \Cref{tab:dependencesensitivity} gives hypothetical lower bounds from \cref{thm:dependentordering} for $T=51{,}833$ and $q=19$. The factor $L_T=Ta_T$ is determined by temporal Gram matrices and, in the stationary case, by the product-lag sum in \cref{eq:dependentlags}. The table uses assumed values of the factor and shows what the independent, rotationally symmetric path model implies; its entries are not confidence bounds for the observed $74.1\%$. Estimates for the competition histories follow below.

\begin{table}[!htb]
\centering\small
\caption{Sensitivity of the northern nonnegative-city preservation lower bound to hypothetical temporal dependence, at $T=51{,}833$ and $q=19$. The general bound is $\max\{0,1-\mathcal B(L_T/T)/\gamma_{19}\}$. The independent-date row uses the sharper \cref{thm:uniformorderingrate}.}
\label{tab:dependencesensitivity}
\begin{tabular}{@{}lrr@{}}
\toprule
Sampling statement & Assumed $L_T$ & Lower bound\\
\midrule
Independent uniform dates & $1$ & $99.01\%$\\
General dependent-history bound & $1$ & $98.75\%$\\
General dependent-history bound & $3.4$ & $97.68\%$\\
General dependent-history bound & $5$ & $97.18\%$\\
General dependent-history bound & $10$ & $96.01\%$\\
General dependent-history bound & $25$ & $93.66\%$\\
General dependent-history bound & $100$ & $87.15\%$\\
\bottomrule
\end{tabular}
\end{table}

Formally replacing $T$ by $T/5=10{,}366.6$ in the independent-date formula gives $97.7853\%$ (the integer count $10{,}366$ gives the same rounded value). This is a sensitivity calculation, not an application to a fractional sample. The replacement is exact for aligned equal blocks of five repeated vectors when $T=5n$, as described in \cref{sec:dependentnull}; a five-hour target does not establish that model. Neither this calculation nor an assumed $L_T=5$ brings the null lower bound down to the observed $74.1\%$. Greater temporal dependence can weaken the probability bound, but cannot remove the null's second-moment restriction.

\paragraph{Estimated lag profiles.}
We estimated the factor for the competition histories in the stationary form of \cref{eq:dependentlags}. For each signal, $\hat\alpha_i(h)$ is the average of $\ip{X_{i,t}}{X_{i,t+h}}$ over row pairs $(t,t+h)$ that lie in the same month block and are both retained in U; no lag crosses a block boundary. The pair factor is
\[
 \hat L_{ij}=1+2\sum_{h=1}^{H}\left(1-\frac hT\right)\hat\alpha_i(h)\hat\alpha_j(h),\qquad T=51{,}833,
\]
a product of the two signals' own lag profiles. Across the 776 signals the median of $\hat\alpha_i(1)$ is $0.469$; the median falls to $0.198$ at lag 5, $0.079$ at lag 24, and $0.041$ at lag 120. \Cref{tab:estimatedlag} gives the distribution of $\hat L_{ij}$ over all 300,700 pairs at four truncations $H$, with the lower bound of \cref{thm:dependentordering} evaluated at the median and at the 95th-percentile factor. At $H=120$ the median factor is $2.4$, with a month-block bootstrap range of $2.3$--$2.5$ (200 replicates), and the implied bound is $98.06\%$; at the 95th-percentile factor of $12$ it is $95.70\%$. Even at $H=720$ and the 95th-percentile factor of $24$, the bound is $93.80\%$. Restricting to both-positive pairs, or centering each history at its mean before computing lags, changes the median factor by less than $0.3$ at every truncation.

\begin{table}[!htb]
\centering\small
\caption{Estimated pair factors $\hat L_{ij}$ from within-block lag profiles (uncentered, all 300,700 pairs in U), and the lower bound $\max\{0,1-\mathcal B(\hat L/T)/\gamma_{19}\}$ of \cref{thm:dependentordering} at $T=51{,}833$. Lags are counted in stored rows.}
\label{tab:estimatedlag}
\begin{tabular}{@{}rrrrr@{}}
\toprule
Truncation $H$ & Median $\hat L$ & 95th-percentile $\hat L$ & Bound at median & Bound at 95th percentile\\
\midrule
$5$ & $1.8$ & $4.1$ & $98.33\%$ & $97.45\%$\\
$24$ & $2.1$ & $7.0$ & $98.20\%$ & $96.66\%$\\
$120$ & $2.4$ & $12$ & $98.06\%$ & $95.70\%$\\
$720$ & $2.9$ & $24$ & $97.84\%$ & $93.80\%$\\
\bottomrule
\end{tabular}
\end{table}

Three qualifications apply. First, lags are counted in stored rows, not in verified elapsed hours: the export has no timestamps, rows were removed, and the 76 month blocks of U hold 682 retained rows on average. Second, the estimate at $h=720$ rests on only 1,032 row pairs, because no lag crosses a block boundary. Third, $\hat L_{ij}$ is a plug-in estimate and can fall below one (its minimum is $0.32$), because products of estimated lag correlations can be negative. The factor is also heterogeneous across pairs: at $H=120$ its 95th percentile is about five times its median. Fourth, the upper tail has not converged: the 95th-percentile factor doubles between $H=120$ and $H=720$, and no lag crosses a month boundary. Given the size of the gap to $74.1\%$, this does not affect the conclusion.

None of these bounds approaches the observed $74.1\%$: the bound of \cref{thm:dependentordering} falls to that level only at a factor of about $390$, against an estimated median of $2.4$ and a 95th percentile of $24$. Within the independent, rotationally symmetric path model, estimated serial dependence therefore does not account for the gap between the null and the data. This is a statement about that model with plug-in lag correlations. It does not turn the null-model bounds into guarantees for the observed library.

\paragraph{A second-moment diagnostic independent of the lag factor.}
For every temporal dependence structure in \cref{thm:dependentordering}, \cref{cor:dependentmoments} gives $q\E D_T^2=a_T\le1$. This remains true after northern selection, although $\E_T^+D_T$ need not be zero. It also holds for mixtures of admissible pair laws with different temporal Gram matrices, because the upper bound holds pair by pair. No analogous invariance is claimed for additional conditioning on city angle or competition-quality filters.

Let $M=300{,}700$, $\overline D=M^{-1}\sum_{i<j}D_{ij}$, and $\overline{D^2}=M^{-1}\sum_{i<j}D_{ij}^2$ denote \emph{empirical} pair averages in U. Jensen's inequality and the reported mean give
\begin{equation}\label{eq:empiricalmomentdiagnostic}
 q\,\overline{D^2}\ge q\,(\overline D)^2
 \simeq19(0.238)^2=1.076236>1.
\end{equation}
Rounding the mean to three decimals does not affect the comparison with one. Thus no admissible choice of temporal covariance matrices can make this null's \emph{expected} second moment equal to the observed empirical second moment. The mismatch is informative as a description: adjusting temporal dependence alone cannot fit that moment while the path law remains independent and rotationally symmetric.

An empirical average is not a model expectation, however, and \cref{eq:empiricalmomentdiagnostic} should be read as a diagnostic only. It does not show that a finite random sample from the null could not produce these data, and it is not a calibrated rejection test: pair comparisons share histories, and competition selection may alter their distribution, so a test would need a sampling and selection model for the whole library. What the diagnostic indicates is that the null's assumptions, taken together, are limited as a description of this library. It cannot, taken alone, tell cross-signal dependence apart from nonuniform directional laws or selection effects. The predominantly positive observed residuals are separate evidence of realized shared movement.

\paragraph{The population margin and its limits of applicability.}
\Cref{thm:populationordering} allows temporal and cross-signal dependence through a jointly stationary ergodic law. For fixed nonzero population means, preservation is governed by the margin $(1-\rho_X\rho_Y)C_*-\kappa$, and shared residual movement can cause persistent violations, as \cref{ex:sharedpersistent} shows. This is a conditional framework for interpreting dependence, not a fitted model for the exported library.

Three observations caution against treating the theorem as a fitted model. The empirical median retained length is only $0.0338$. Full-sample preservation varies from approximately $42\%$ to $64\%$ across the reported windows. And month resampling shifts its median down by $4.3$ percentage points. None of these disproves fixed nonzero means or stationarity; equally, none shows that city directions are already close to stable population directions. If the means are near zero relative to sampling noise, or if they drift, the fixed-mean asymptotic description can be inadequate. \Cref{prop:localalternatives} illustrates a mean--noise transition under specified iid assumptions; extending that transition law to dependent or drifting histories requires further assumptions. Until population stability is established, the finite-sample residual identity is the direct explanation of the measured ordering.

\subsection{Robustness, uncertainty, and export limitations}\label{sec:robustness}

A second sample P retains all 834 long histories and all 56,120 common rows as stored, including zeros and non-unit rows. Its raw stored inner-product ordering gives preservation rates of $59.9\%$ overall, $77.9\%$ for $C_{ij}\ge0$, and $2.0\%$ for $C_{ij}<0$. Sample P differs from U in signal membership, support, and row treatment at once, so this is a joint sensitivity check. Because P does not satisfy the unit-vector condition, its rates do not have the spherical-angle interpretation of those for U, and the difference between the two samples cannot be attributed to zero rows alone.

Keeping duplicate submissions gives $58.6\%$ overall preservation; excluding the 17 pairs identical on U's support, or the 79 pairs with both $C_{ij}>0.999$ and $D_{ij}>0.999$, leaves the rounded $58.5\%$ rate unchanged. Restricting to the 763 unit-compatible histories covering all 80 reference months gives $57.4\%$. Same-scientist and different-scientist pairs preserve separation in $67.4\%$ and $58.4\%$ of cases, respectively, with the latter accounting for 298,224 of the 300,700 pairs.

Cities and all pair statistics were also recomputed within time windows. Across the reported disjoint partitions the overall rate ranges from approximately $42\%$ to $64\%$, while the nonnegative-city rate ranges from approximately $65\%$ to $80\%$ and the negative-city rate stays below $4\%$. The share of nonnegative-city pairs changes substantially across windows. On the last retained month alone, preservation is $38.9\%$ and median retained length is $0.169$. These comparisons show that the rates are sensitive to the evaluation period. Because period and window length change together, they do not isolate a causal effect of length.

A complementary bootstrap resamples 152 scientist clusters while holding pair statistics fixed. With 1,000 replicates, its overall angle-preservation range is $55.1\%$--$62.3\%$. The month-block analysis in \cref{sec:screenuncertainty} instead recomputes cities and histories, producing a wider $47.5\%$--$61.6\%$ range and the downward median shift reported there. The two exercises answer different questions.

The same cluster draws give ranges for threshold screening and for the positive-IC subset (\cref{tab:clusterbootstrap}). A pair of distinct signals is counted $w_aw_b$ times, where $w$ is the number of drawn copies of each signal; a signal is never paired with its own copy, and same-scientist pairs of distinct signals are kept. At the primary threshold the improvement over the base rate is $9.3$ percentage points, with a cluster-bootstrap range of $7.2$--$12.1$. A two-way bootstrap, which resamples scientist clusters and month blocks in the same replicate and recomputes cities, gives $7.1$--$12.2$. Both ranges exclude zero and are wider than the month-block range of $8.2$--$10.0$. At $\tau=-0.5$ the range reaches zero: only $0.07\%$ of pairs have $D\le-0.5$, and in some replicates none is drawn. For both-positive pairs, the cluster bootstrap is centered on the original values, because pair statistics are held fixed, while the two-way bootstrap inherits the downward shift of month resampling.

\begin{table}[!htb]
\centering\small
\caption{Scientist-cluster resampling (152 clusters, pair statistics fixed) and two-way resampling of clusters and month blocks (cities recomputed), 1,000 replicates each. Screening rows without a stated threshold use $c=\tau=0.5$. Cells give the bootstrap median and the 2.5--97.5 percentile range; improvements are in percentage points.}
\label{tab:clusterbootstrap}
\begin{tabular}{@{}lrrr@{}}
\toprule
Statistic & Original & Cluster bootstrap & Two-way bootstrap\\
\midrule
Pairs accepted & $72.5\%$ & $72.4\%$ ($67.2$--$77.2$) & $74.0\%$ ($67.6$--$79.0$)\\
Success among accepted & $94.3\%$ & $94.4\%$ ($91.8$--$96.2$) & $94.2\%$ ($91.1$--$96.3$)\\
Base rate & $85.0\%$ & $84.8\%$ ($81.0$--$88.4$) & $84.9\%$ ($80.3$--$88.6$)\\
Improvement & $+9.3$ & $+9.5$ ($+7.2$ to $+12.1$) & $+9.2$ ($+7.1$ to $+12.2$)\\
Accepted given $D>0.5$ & $27.5\%$ & $27.3\%$ ($20.7$--$33.6$) & $28.5\%$ ($20.3$--$36.4$)\\
Rejected given $D\le0.5$ & $19.6\%$ & $19.4\%$ ($15.4$--$24.2$) & $17.8\%$ ($13.8$--$23.5$)\\
\midrule
Improvement, $c=\tau=0.75$ & $+1.8$ & $+1.9$ ($+1.3$ to $+2.8$) & $+1.9$ ($+1.3$ to $+2.9$)\\
Improvement, $c=\tau=0.25$ & $+19.4$ & $+19.2$ ($+15.9$ to $+23.3$) & $+18.2$ ($+13.2$ to $+24.3$)\\
Improvement, $c=\tau=0$ & $+10.0$ & $+9.9$ ($+6.9$ to $+13.4$) & $+11.5$ ($+4.1$ to $+23.3$)\\
Improvement, $c=\tau=-0.5$ & $+3.6$ & $+3.3$ ($+0.0$ to $+11.8$) & $+3.0$ ($+0.0$ to $+11.9$)\\
\midrule
Both-positive: preserved & $58.9\%$ & $59.0\%$ ($55.4$--$63.0$) & $54.7\%$ ($47.6$--$62.4$)\\
Both-positive: preserved, $C\ge0$ & $74.1\%$ & $74.2\%$ ($70.7$--$77.9$) & $71.9\%$ ($66.5$--$77.3$)\\
Both-positive: preserved, $C<0$ & $1.5\%$ & $1.5\%$ ($0.9$--$2.3$) & $1.0\%$ ($0.3$--$3.0$)\\
\bottomrule
\end{tabular}
\end{table}

These exercises remain indicative. None models dependence between successive months, and the two-way scheme resamples scientists and months independently of each other. That a majority of pairs preserve separation is a fact about this sample, not a robust lower bound across periods: the month-block and two-way ranges for overall preservation both extend below one half.

Direct recomputation of 200 sampled pairs agreed with batched values of $D$, $C$, and $K$ to $10^{-14}$ or better. No residual-energy bound violations were observed beyond floating-point rounding. Recomputed full-history ICs, lengths, and cities also match the supplied summary tables to approximately $10^{-14}$. These checks establish that the analysis of the export is numerically consistent. They do not independently validate how the gauge frames were constructed upstream.

Three export limitations remain material. First, neither per-row timestamps nor explicit month labels accompany the stored blocks: chronology is inferred and cross-checked against local non-production data. Second, adjacent blocks overlap by six row labels at most boundaries (four at three boundaries). Without timestamps we cannot tell whether these labels represent the same hours twice, so every observation keeps the weight of its production block. Third, the five shared non-unit rows per month occur at month end, but their upstream cause has not been verified against the production gauge transforms. Removing those rows establishes the norm condition on U, but does not itself verify orthogonality and a common frame against the original signals. The findings consequently describe the cleaned exported histories. Production timestamps, month metadata, and direct checks of the common gauge transforms would complete the provenance validation.

The empirical contribution has three parts. The first is the measured frequency with which the city angle is conservative among positive-IC pairs, in particular the $74.1\%$ rate for city angles of at most $90^\circ$. The second is the gain from threshold screening over the base rate at each threshold, together with its cost in rejected pairs. The third is the temporal resolution at which block-mean certificates begin to decide pairs (\cref{sec:blockexperiment}). All three are pairwise, retrospective comparisons on the same evaluation support. The first two are signed and do not evaluate an absolute-correlation rule $|D|\le\tau$. None of the three evaluates a simultaneous guarantee for a selected signal portfolio or performance on future periods; each of those applications requires a separate evaluation. The threshold comparisons are descriptive; the cutoffs were not optimized out of sample. And an inconclusive deterministic certificate still does not establish duplication.

\Needspace{10\baselineskip}
\section{Conclusion}\label{sec:conclusion}

We asked when the separation between two cities is a conservative estimate of the separation between the underlying histories. In the observed Competition May 2026 library, positive-IC pairs preserve separation in $58.9\%$ of cases overall and in $74.1\%$ when their city angle is at most $90^\circ$: in that regime, about three quarters of pairs have histories at least as separated as their cities. Beyond $90^\circ$, preservation is $1.5\%$. The complement, in which at least one signal lacks positive average IC, preserves separation in $55.3\%$ of cases overall and in $92.6\%$ when $C\ge0$. Positive residual agreement is more frequent among both-positive pairs ($91.0\%$ against $70.4\%$). These contrasts are a descriptive association with the sign of average IC. They do not identify a causal effect, and they do not make the northern-null benchmark applicable to the complement. The results by evaluation window and under month resampling show that the overall majority should not be extrapolated to every period: the month-block range for both-positive pairs is $48.3$--$61.2\%$ overall, against $68.2$--$75.7\%$ for city angles of at most $90^\circ$.

The uniform null describes an idealized mechanism. Under it, history correlation approaches zero, so the sign of the city correlation alone determines the limiting ordering. Selecting for positive IC changes the aggregate limit to $\gamma_q$, approximately $56\%$ for 20 assets; the conditional limits of one and zero hold with or without that selection. \Cref{thm:dependentordering} allows temporal dependence within independent rotationally symmetric histories and keeps those limits whenever its temporal Gram factor $a_T$ tends to zero. Serial dependence therefore need not destroy the asymptotic result, although it does change the finite-sample bounds.

The empirical frequencies do not validate this null. Lag profiles estimated from the stored histories give a median pair factor of about $2.4$, and null lower bounds between $93.8\%$ and $98.3\%$ across the truncations and percentiles reported, far above the observed $74.1\%$. Moreover, no choice of temporal covariance within the null can match the empirical second moment in \cref{eq:empiricalmomentdiagnostic}. That mismatch is a descriptive diagnostic, not a calibrated rejection test for a finite library. In the data, shared residual movement determines the observed margin, and \cref{thm:populationordering} gives the limiting counterpart of that margin for jointly stationary ergodic histories with fixed nonzero means. Because retained lengths are small and the rates are sensitive to the evaluation window, those stability assumptions are not established here; a mean--noise transition, or population means that change over time, remain possible.

The threshold study answers the pairwise screening question directly. At a $60^\circ$ cutoff on the full 776-signal sample, the city screen retains $72.5\%$ of pairs and $80.4\%$ of those meeting the history threshold. Its success rate is $94.3\%$, against a base rate of $85.0\%$, and its improvement range is $8.2$--$10.0$ percentage points under month-block resampling, $7.2$--$12.1$ under scientist-cluster resampling, and $7.1$--$12.2$ when both are resampled together. At this fixed target, tighter city cutoffs add little success and discard many acceptable pairs. Two gaps remain. The provenance of the timestamps and of the gauge transforms is incomplete, so lags are measured in stored rows. And prospective screening performance is not established.

The deterministic results show what city directions and retained lengths can guarantee for arbitrary histories. The bounds are sharp, and for this sample's small retained lengths they are inconclusive at the $0.5$ threshold. Means over shared time blocks give nested certificates that tighten as the blocks are refined and recover the full-history correlation at the finest resolution. In this library they need fine resolution: at the $0.5$ threshold, fewer than $0.05\%$ of the pairs below the threshold are certified with 512 blocks or fewer, $54.7\%$ when each signal stores about a third as many vectors as its history has rows, and $83.4\%$ at about two thirds. Persistent positive IC confines a city and bounds its displacement from one observation to the next; counterexamples show that confinement does not imply convergence. Storing the mean vector rather than the city alone, and finer block means when needed, keeps explicit the information on which these guarantees depend.

\Needspace{10\baselineskip}
\phantomsection
\end{document}